\documentclass[12pt]{article}
\usepackage[dvipsnames,svgnames,x11names,hyperref]{xcolor}
\usepackage{amsmath}
\usepackage{amsthm}
\usepackage{xspace}
\usepackage{float}
\usepackage[T1]{fontenc}
\usepackage{mdframed}

\usepackage{fullpage}
\usepackage{thmtools}
\usepackage[colorlinks]{hyperref}
\hypersetup{colorlinks={true},linkcolor={blue},citecolor=magenta}

\usepackage{amssymb,amsfonts,amsmath,amsthm,amscd,dsfont,mathrsfs}
\usepackage{complexity}
\usepackage{tikz}
\usetikzlibrary{decorations.pathreplacing,backgrounds}
\usepackage{multirow}
\usepackage{mathpazo}
\usepackage{braket}
\usepackage{todonotes}

\usepackage[
backend=biber,
style=alphabetic,
sorting=ynt
]{biblatex}
\usepackage{url}
\usepackage{bbm}
\usepackage{cleveref}
\usepackage{hyperref} 
\hypersetup{breaklinks=true}

  \theoremstyle{plain}
  \newtheorem{theorem}{Theorem}[section]
  \newtheorem{lemma}[theorem]{Lemma}
  \newtheorem{corollary}[theorem]{Corollary}
  \newtheorem{definition}[theorem]{Definition}
  \newtheorem{remark}[theorem]{Remark}
  
    \newtheorem{fact}[theorem]{Fact}
\newtheorem{example}{Example}[section]

 \newtheorem*{theorem*}{Theorem}

\newmdtheoremenv[backgroundcolor=gray!10,
                 linewidth=0pt,
                 innerleftmargin=4pt,
                 innerrightmargin=4pt,
                 innertopmargin=1pt,
                 innerbottommargin=4pt,
            splitbottomskip=4pt]{problem}[prob]{Problem}

\newmdtheoremenv[backgroundcolor=gray!10,
                 linewidth=0pt,
                 innerleftmargin=4pt,
                 innerrightmargin=4pt,
                 innertopmargin=1pt,
                 innerbottommargin=5.5pt,
            splitbottomskip=4pt]{conjecture}[conj]{Conjecture}

\newcommand{\N}{\mathbb{N}}

\newcommand{\Id}{\mathbb{1}}

\newcommand{\proj}[1]{\ensuremath{|#1\rangle \langle #1|}}

\newcommand{\Tr}{\mathrm{Tr}}

\renewcommand{\E}{\mathbb{E}}

\newcommand{\bit}{\{0,1\}}

\newcommand{\algo}{\mathcal}
\newcommand{\from}{\ensuremath{\leftarrow}}
\newcommand{\negl}{\ensuremath{\operatorname{negl}}\xspace}
\newcommand\numberthis{\addtocounter{equation}{1}\tag{\theequation}}

\title{Quantum One-Wayness of the Single-Round Sponge with Invertible Permutations}

\author{Joseph Carolan\footnote{\texttt{jcarolan@umd.edu}}\\University of Maryland \and Alexander Poremba\footnote{\,\texttt{poremba@mit.edu}}\\MIT}

\date{}

\begin{document}

\maketitle


\abstract{
Sponge hashing is a widely used class of cryptographic hash algorithms which underlies the current international hash function
standard SHA-3. In a nutshell, a sponge function takes as input a bit-stream of any length 
and processes it via a simple iterative procedure: it repeatedly feeds each block of the input into a so-called \emph{block function}, and then produces a digest by once again iterating the block function on the final output bits.
While much is known about the post-quantum security of the sponge construction when the block function is modeled as a random function or one-way permutation, the case of \emph{invertible} permutations, which more accurately models the construction underlying SHA-3, has so far remained a fundamental open problem.

In this work, we make new progress towards overcoming this barrier and show several results. First, we prove the ``double-sided zero-search'' conjecture proposed by Unruh (eprint' 2021) and show that finding zero-pairs in a random $2n$-bit permutation requires at least $\Omega(2^{n/2})$ many queries---and this is tight due to Grover's algorithm. At the core of our proof lies a novel ``symmetrization argument'' which uses insights from the theory of Young subgroups. Second, we consider more general variants of the double-sided search problem and show similar query lower bounds for them. As an application, we prove the quantum one-wayness of the single-round sponge with invertible permutations in the quantum random permutation model.
}

\section{Introduction}

Hash functions are one of the most fundamental objects in cryptography. They are used in a multitude of applications, such as integrity checks for data packages (e.g., software updates), password storage on cloud servers, or as important components in digital signature schemes---whether it is to design hash-based signatures or simply to construct signatures schemes for variable input lengths~\cite{10.5555/2700550}.

In recent years, the National Institute of Standards and Technology (NIST) announced a new international hash function standard known as SHA-3. Unlike its predecessor SHA-2, which was rooted in the Merkle-Damg\aa rd construction~\cite{Mer88,Mer90,eurocrypt-1987-2247}, the new hash function standard uses Keccak~\cite{KeccakSub3}---a family of cryptographic functions based on the idea of \emph{sponge hashing}~\cite{KeccakSponge3}. This particular approach  allows for both variable input length and variable output length, which makes it particularly attractive towards the design of cryptographic hash functions. 
The internal state of a sponge function gets updated through successive applications of a so-called \emph{block function} $\varphi: \bit^{r+c} \rightarrow \bit^{r+c}$, where we call the parameters $r \in \N$ the \emph{rate} and $c \in \N$ the \emph{capacity} of the sponge. 
The evaluation of the sponge function $\mathsf{Sp}^\varphi: \bit^* \rightarrow \bit^*$ takes place in two phases:
\begin{enumerate}
    \item (Absorption phase:) During each round, a new block consisting of $r$ many bits of the input gets ``fed'' into the block function $\varphi$. Following the sponge metaphor, we say that the sponge function ``absorbs'' the input. 
    More formally, suppose that the input consists of $r$-bit blocks $m_1,m_2,\dots,m_\ell$. During the first round, we compute $\varphi(m_1||0^c) = y_1||z_1$, for some $y_1 \in \bit^r$ and $z_1 \in \bit^c$. Then, to absorb the next block $m_2$, we compute $\varphi(y_1 \oplus m_2||z_1) = y_2 || z_2$, and so on.

    \item (Squeezing phase:) During each round, a new $r$-bit block is produced by
    essentially running the absorption phase in reverse. Each $r$-bit digest becomes a block of the final output of the function $\mathsf{Sp}^\varphi$.
    Following the sponge metaphor, we say that
the sponge function gets “squeezed” to produce fresh random bits.
More formally, suppose the last round of the absorption phase results in a pair $y_\ell||z_\ell$. We then let the digest $y_\ell$ serve as the first block of the output. To produce the second block, we compute $\varphi(y_\ell||z_\ell)=y_{\ell+1} || z_{\ell+1}$, and output $y_{\ell+1}$, and so on.
\end{enumerate}

Several works have since analyzed the security properties behind the sponge construction~\cite{10.1007/978-3-540-78967-3_11,10.1007/978-3-031-15982-4_5,10.1007/978-3-031-48621-0_9}. In particular, in~\cite{10.1007/978-3-540-78967-3_11} it was shown that the sponge hash function enjoys a strong form of security called \emph{indifferentiability} in the case when the underlying block function $\varphi$ is modeled as an invertible random permutation. 
Indifferentiability already implies many desired properties of cryptographic hash functions, such as collision-resistance or pseudorandomness.

While much is known about the security of the sponge in a classical world, our understanding changes significantly once we take the threat of large-scale quantum computing into account. The good news is that most hash functions are believed to only be mildly affected, whereas the majority of pre-quantum public-key cryptography faces serious threats due to 
Shor's algorithm (and its variants)~\cite{Shor_1997,regev2024efficient}. The reason hash functions are believed to be less susceptible to quantum attacks is due to their inherent lack of structure, which means that generic quantum attacks tend to achieve at most a square-root speed-up relative to their classical counterparts. 
To this day, however, our understanding of the post-quantum security of the sponge remains somewhat incomplete.
In fact, the post-quantum security of the sponge construction is only well-understood in the special case when the block function is modeled as a \emph{non-invertible} random permutation~\cite{Czajkowski2017,Czajkowski2019}. The case of invertible random permutations, which more accurately models the construction underlying SHA-3, has so far remained a major open problem.


\paragraph{Single-round sponge hashing.}

In the special case when there is only a single round of absorption and squeezing, the sponge function $\mathsf{Sp}^\varphi: \bit^{r} \rightarrow \bit^r$ has a simple form which is illustrated in Figure~\ref{fig:single-sponge}; namely, on input $x \in \bit^r$, the output is given by $y = \mathsf{Sp}^\varphi(x)$, where $y$ corresponds to the first $r$ bits of $\varphi(x||0^c)$. In other words, $\mathsf{Sp}^\varphi$ is defined as the restriction of $\varphi$ onto the first $r$ bits of its output.

\begin{figure}[t]
\begin{center}
{\small
\begin{tikzpicture}
  \draw (5,-0.75) rectangle (6,0.75) node [pos=.5]{$\varphi$}; 

\draw[-] (4.6,0.35) node[left]{$x$ \hspace{1mm}} -- (5,0.35);
\draw[-] (4.6,-0.35) node[left]{$0^c$} --(5,-0.35);
  \draw[-] (6.0,0.35) node[right]{\hspace{5mm}$y$} --(6.4,0.35);
\draw[-] (6.0,-0.35) node[right]{\hspace{5mm}$z$} --(6.4,-0.35);

 \end{tikzpicture}

\label{fig:single-sponge}
}
\end{center}
\caption{The single-round sponge.}
\end{figure}
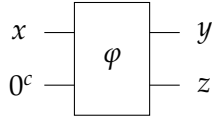

Despite its simplicity, the single-round sponge already features many of the technical difficulties that one encounters when analyzing the security of the (many-round) sponge---especially once we model $\varphi$ as an invertible random permutation. 
In the post-quantum setting, this requires us to analyze quantum algorithms $\algo A^{\varphi,\varphi^{-1}}$ which are allowed to make superposition queries to both the ``forward oracle'' $O_\varphi$, as well as the ``backward oracle'' $O_{\varphi^{-1}}$.
For example, one may ask:
\begin{itemize}
    \item (One-wayness:) How many queries does a quantum algorithm $\algo A^{\varphi,\varphi^{-1}}$ need to invert $\mathsf{Sp}^\varphi$ on a randomly chosen input? In other words, given $y$ which corresponds to the first $r$ bits of $\varphi(x||0^c)$, for a random $x \sim \bit^r$, how many queries does $\algo A$ need to find $x'$ such that the first $r$ many bits of $\varphi(x'||0^c)$ equal $y$? 

    \item (Collision-resistance:) How many queries does a quantum algorithm $\algo A^{\varphi,\varphi^{-1}}$ need to find a collision for $\mathsf{Sp}^\varphi$? Here, $\algo A$ needs to find a colliding pair $x,x'$ such that the first $r$ many bits of $\varphi(x||0^c)$ and $\varphi(x'||0^c)$ match.
\end{itemize}
To the best of our knowledge, there are currently no known (tight) query lower bounds for either of the two above problems---despite the fact that they concern the seemingly simple variant of single-round sponge hashing. 
While prior work by Zhandry~\cite{Zhandry21} does in fact show the strong notion of reset-indifferentiabability of the single-round sponge with invertible permutations (which implicitly implies the hardness of the query problems above), the techniques do not seem to yield tight lower bounds and only apply to restricted parameter ranges (see \Cref{sec-related}).
The lack of lower bounds is mainly due to the fact that current proof techniques for the non-invertible case~\cite{Czajkowski2017,Czajkowski2019} seem to break down once the inverse oracle $O_{\varphi^{-1}}$ enters the picture.

Consider, for example, the basic property of \emph{zero-preimage-resistance}; that is, the hardness of finding a pre-image of the all-zero string. Czajkowski et al.~\cite{Czajkowski2017} showed that if a (non-invertible) $\varphi$ is both collision-resistant (when restricted to the left and right half of its
output) and zero-preimage-resistant (when restricted to the right half of its output),
then the (many-round) sponge construction is collision-resistant. Unfortunately, 
this approach fails immediately if an
adversary can submit inverse queries to the block function $\varphi$. Clearly, the restriction of $\varphi$ to the right half of its output is no longer zero-preimage-resistant---an adversary can simply query $O_{\varphi^{-1}}$ on inputs of the form $y||0^c$, for any input $y \in \bit^r$, in order to find a zero-preimage. 

Another difficulty in generalizing existing proof techniques lies in the fact it is currently not known how to construct \emph{compressed permutation oracles}~\cite{Unruh2021,Unruh2023}.
Zhandry~\cite{Zhandry2018} previously introduced compressed oracles as a tool to prove query lower bounds for problems involving random oracles; specifically, as a means to ``record'' quantum queries to a random oracle.
Compressed oracles have proven to be extremely useful in analyzing the post-quantum security of hash functions in the quantum random oracle model, particularly in the case of random sponges~\cite{Czajkowski2019,Unruh2021,Unruh2023}. However, unlike random functions where all the function outputs are sampled independently, this is not true for random permutations. This makes it difficult to construct compressed oracles for permutations, especially for modelling algorithms that can query the permutation in both directions~\cite{Unruh2021,Unruh2023}.

\paragraph{Double-sided zero-search.}

In light of the difficulty in proving the post-quantum security of the sponge with invertible permutations,
Unruh proposed the following simple conjecture which seems beyond the scope of current techniques.

\begin{conjecture}[Double-sided zero-search, \cite{Unruh2021,Unruh2023}]\label{conj}
Any quantum algorithm $\algo A^{\varphi,\varphi^{-1}}$ which has quantum query-access to a random permutation $\varphi: \bit^{2n} \rightarrow \bit^{2n}$ and its inverse $\varphi^{-1}$ must make $\Omega(2^{n/2})$ many queries to find a zero pair $(x,y)$ such that $\varphi(x||0^n) = y||0^n$ with constant success probability.
\end{conjecture}
Note that the above lower bound is already tight---a simple Grover search allows one to find a zero-pair using $\Theta(2^{n/2})$ many queries, if one exists. Notice that there is a slight subtlety that arises in the statement of \Cref{conj} since not every permutation has a 
zero pair; for example, the permutation $\varphi(x||y) := (x||y \oplus 1^n)$ has no zero pairs since $\varphi(x||0^n) = x||1^n$, for any $x \in \bit^n$, whereas the identity permutation $\varphi := \mathrm{id}$ has precisely $2^n$ zero pairs. However, for a random permutation there is at least one zero pair with constant probability---a property we show in \Cref{fact:Z-pairs-existence-probability}.

While it is not immediately clear how the double-sided zero-search problem from \Cref{conj} relates to the security of the sponge construction, it was suggested by Unruh~\cite{Unruh2021} that it may already provide some evidence for its collision-resistance. The problem has so far also resisted attempts\footnote{This was pointed out by Unruh~\cite{Unruh2021}.} at solving it using
standard techniques from quantum query complexity; for example, the adversary method~\cite{AMBAINIS2002750} or the polynomial method~\cite{10.1145/502090.502097}. Therefore, a resolution to \Cref{conj} may already offer interesting new insights into proving the post-quantum security of the sponge construction with invertible permutations.

\subsection{Our contributions}

We now give an overview of our contributions in this work.

\paragraph{Resolving Unruh's conjecture.}

We prove the double-sided zero-search conjecture (\Cref{conj}) due to Unruh~\cite{Unruh2021,Unruh2023} and show that finding zero-pairs in a random $2n$-bit permutation with constant probability requires at least $\Theta(2^{n/2})$ many queries. Specifically, we show the following in \Cref{thm:uniformZeroSearchHard}.
\begin{theorem*}[Informal]
    Any quantum algorithm for \textsc{Double-Sided Zero-Search} that makes $T$ queries to an invertible permutation succeeds with probability at most
    $O(T^2/2^n)$.
\end{theorem*}
In other words, any algorithm must make at least $T = \Omega(\sqrt{\epsilon 2^n})$ many queries in order to succeed with probability $\epsilon$. As we observe in \Cref{cor:grover-dszs}, this immediately yields a tight lower bound of $\Theta(2^{n/2})$ for constant success probability due to Grover, or more generally a tight bound of $\Theta(\sqrt{\epsilon 2^n})$ for any success probability $\epsilon > 0$. 

Our proof takes place in two steps:

\begin{itemize}
    \item \textbf{Reduction from worst-case unstructured search:} 
    Our first insight lies in the fact that we can reduce \textsc{Unstructured Search} with $K$ out of $2^n$ marked elements to a specific instance of \textsc{Double-Sided Zero-Search} without any overhead; namely, one in which the $2n$-bit permutation $\varphi$ has exactly $K$ zero pairs. Suppose we have quantum oracle access to a function $$f:\{0,1\}^n \rightarrow \{0,1\} \quad\quad\text{ with } \quad\quad |f^{-1}(1)|=K.$$ 
    We now construct the following permutation $\varphi:\{0,1\}^{2n} \rightarrow \{0,1\}^{2n}$ such that, for any inputs $x, y \in \{0,1\}^n$, it holds that:
    $$
    \varphi(x||y) =\begin{cases}
        x||y & \text{ if } f(x)=1 \\
        x||(y \oplus 1^n) & \text{ if } f(x)=0 \, .
    \end{cases}
    $$
    Notice that while $\varphi$ is not random, it has exactly $K$ zero pairs on inputs of the form $x||0^n$ whenever $x$ satisfies $f(x)=1$. Furthermore, $\varphi=\varphi^{-1}$ and so backwards queries are not helpful---thus worst case hardness follows.
    
    \item \textbf{Worst-case to average-case reduction via symmetrization:} 
    Our next insight lies in the fact that we can re-randomize any worst-case permutation $\varphi$ with $K$ zero pairs into an average-case permutation $\varphi^{\mathrm{sym}}$ with $K$ zero pairs: \begin{enumerate}
        \item Select two $2n$-bit permutations $\sigma,\omega$ such that both preserve the property of ending in $0^n$, but are otherwise independently random.\footnote{More formally, we require that $\sigma$ and $\omega$ map strings of the form $(x||0^n)$ to $(y||0^n)$, for $x,y \in \bit^n$.}
        \item Symmetrize $\varphi$ by letting $\varphi^{\mathrm{sym}} = \omega \circ \varphi \circ \sigma$.
    \end{enumerate}
    A zero pair of $\varphi$ can be reconstructed from a zero pair of $\varphi^{\mathrm{sym}}$, but the location of the zero pairs and the permutation behaviour everywhere else is now randomized. We have inverse access to $\varphi$, and so we also have inverse access to $\varphi^{\mathrm{sym}}$. Hardness of the average case is now implied by hardness of the worst case.
\end{itemize}
Our symmetrization argument (formally shown in \Cref{lemma:rerandomizationZeroSearch}) uses insights from the theory of \emph{Young subroups}. We give an extensive treatment of the subject in \Cref{sec:subsetPairsSymGroup}.

In \Cref{sec:alternative}, we give an alternative proof of Unruh's double-sided zero-search conjecture (\Cref{conj}) which combines our technique of symmetrization with a more conventional approach rooted in one-way to hiding~\cite{cryptoeprint:2018/904}. The advantage of our alternative proof is that it is more direct---it avoids the two-step template of the previous proof.
At a high level, it uses the superposition oracle framework~\cite{AMBAINIS2002750}
and introduces a \emph{function register} which is outside of the view of the query algorithm, and contains a uniform superposition over permutations. Note that the superposition oracle formalism makes it especially easy to analyze quantum query-algorithms which query permutations both in the forward, as well as in the backward direction.

Once we introduce the function register, we can then ``symmetrize it in superposition''. This framework allows us to directly show that no quantum algorithm can distinguish whether it is querying a random invertible permutation with exactly $K$ zero pairs, or a random
invertible permutation with no zero pairs—unless it makes a large number of queries. 
We remark, however, that contrary to the work of Zhandry~\cite{Zhandry2018}, we do not need to ``compress'' the superposition oracles and use inefficient representations instead.

\paragraph{Subset pairs and double-sided search.}

Motivated by the sponge construction, we generalize these techniques to handle search problems which are non-uniform in the sense of (1) different constraints on preimages vs images and (2) non-uniform distribution on permutations. In particular, we define a variant of two-sided search in which one is given query access to an $n$-bit permutation $\varphi:\bit^n \rightarrow \bit^n$ and its inverse, for an integer $n=r+c$, and asked to produce inputs that end in $c$ many zeros and whose outputs begin in $r$ many zeros. We further consider a class of non-uniform distributions on permutations that weights permutations according to their number of solutions. Notably, under this new distribution a solution is always guaranteed to exist, avoiding a subtlety in the original uniform problem. We show that these modifications do not make the problem significantly easier in \Cref{thm:nonUniformZeroSearchHard}. 
\begin{theorem*}[Informal]
    Any quantum algorithm for \textsc{Non-uniform Double-Sided Search} that makes at most $T$ queries to an invertible (non-uniform) random permutation and succeeds with probability $\epsilon>0$ satisfies
    $\epsilon = O(T^2/2^{\min(r, c)})$.
\end{theorem*}

The argument mirrors that of the uniform case, though with a slightly more complicated worst-case instance.

\paragraph{Quantum one-wayness of the single-round sponge.}

As an application of our techniques, we give a reduction from our \textsc{Non-uniform Double-Sided Search} problem to the one-wayness (formally defined in \Cref{def:one-wayness}) of the single-round sponge. As a corollary of our bound for \textsc{Non-uniform Double-Sided Search}, in \Cref{thm:spongeOneWay} we establish the first arbitrary parameter post-quantum security result for sponge hashing in the random permutation model.

\begin{theorem*}[Informal]
    Any $T$-query quantum algorithm that breaks the quantum one-wayness of the single-round sponge, where the block function is instantiated with an invertible random permutation, has a success probability of at most $\epsilon = O(T^2/2^{\min(r, c)})$.
\end{theorem*}
By modeling the block function as an invertible random permutation, our theorem shows the post-quantum security of the single-round sponge in the quantum random oracle model~\cite{10.1007/978-3-642-25385-0_3}.
To prove our main theorem, we first switch to an alternative but equivalent characterization of one-wayness of the sponge (\Cref{lem:oneWayAlternateGameEqual}), where the image is chosen uniformly at random before the actual block function is sampled (in this case, $\varphi$). Note that this approach, however, results in a non-uniform distribution over $\varphi$. To complete the proof, we then give a reduction from the \textsc{Non-Uniform Double-Sided Search} problem (formally defined in \Cref{def:one-wayness}).

\paragraph{Combinatorics of subset-pairs.}

Along the way, we work out the combinatorics of the expected number of input-output pairs of any given kind (i.e., for a general \emph{subset-pair}), as well as strong tail bounds in the case where the expected number is small. We show that the number of such pairs decays exponentially in this case for both the uniform distribution in \Cref{thm:X-pairs-uniform-tail-bound} and a certain non-uniform one in \Cref{thm:X-pairs-nonuniform-tail-bound}. 

\begin{theorem*}[Informal]
    Whenever the average number of subset pairs of a random permutation is a constant, the probability that there are more than $K$ subset pairs is at most $\exp(-\Omega(K))$.
\end{theorem*}

In the uniform case, the number of subset pairs of permutations follows a so-called \emph{hypergeometric distribution}---a well studied distribution corresponding to sampling without replacement. We relate our non-uniform distribution, in which permutations are favored proportional to their number of subset pairs, to the uniform distribution, and show how to derive strong tail bounds in this case as well.

\subsection{Open questions}


This work does not address post-quantum security of the many-round sponge, which is a natural and important next direction. In particular, we do not address whether the many-round sponge when instantiated with an invertible random permutation is one-way. 
However, we believe that our techniques could potentially open up new directions towards proving the post-quantum security of the many-round sponge, especially when combined with insights from the non-invertible case~\cite{Czajkowski2017}.

Another natural question is whether one can show stronger security properties of the sponge with invertible permutations, even in the single-round case. For example, whether the single-round sponge is \emph{collision-resistant} (or more generally, \emph{collapsing}) when the block function is instantiated with an invertible random permutation. This is not addressed by our work. However, in this case, it is conceivable that there exist direct reductions from collision-resistance (or collapsing) to natural query problems involving invertible permutations, similar to our reduction for one-wayness. In such a case, our techniques would be well suited to analyzing this type of problem.

\subsection{Related work}
\label{sec-related}

Czajkowski, Bruinderink, Hülsing, Schaffner and Unruh~\cite{Czajkowski2017} showed that the many-round sponge is \emph{collapsing}---a strengthening of collision-resistance---in the case when the block function is modeled as a random function or a random (non-invertible) permutation. 
In the same model, Czajkowski, Majenz, Schaffner and Zur~\cite{Czajkowski2019} proved the quantum indifferentiability of the (many-round) sponge. However, in contrast with the former work, the latter relies on Zhandry's compressed oracle technique~\cite{Zhandry2018}. Inspired by Zhandry's compressed oracle method, Rosmanis~\cite{rosmanis2022tight} recently also gave tight bounds for the permutation inversion problem. We remark that the results in the aforementioned works do not apply in the case where the adversary has access to an invertible random permutation.
Zhandry~\cite{Zhandry21} showed that the single-round sponge (in the special case when the message length is roughly half the block length of the permutation) is quantumly indifferentiable from a random oracle (even if the adversary has access to the inverse of the permutation). Because indifferentiability implies that the function behaves ``just like a random oracle'', this immediately implies many desirable properties for this particular  parameter range, such as one-wayness and collision-resistance.
Our results, however, are incomparable to those of Zhandry. While~\cite{Zhandry21} obtains a much more general result, the advantage of our techniques is that we do not require any constraints on the message/block length. Moreover, the restriction to half the block length in~\cite{Zhandry21} is necessary, as the strong form of \emph{reset-indifferentiability} which is achieved is not possible if the message exceeds half the block length. Finally, we remark that
our techniques yield tight lower bounds for the one-wayness of the single-round sponge, whereas the techniques in~\cite{Zhandry21} seem insufficient for this purpose.
Unruh proposed the notion of
\emph{compressed permutation oracles}~\cite{Unruh2021,Unruh2023} which are aimed at dealing with invertible permutations. However, the soundness of the construction is still unresolved and currently relies on a conjecture. Unruh also put forward the double-sided zero-search conjecture~\cite{Unruh2021,Unruh2023} as a simple and non-trivial query problem in the context of invertible random permutations which may potentially have implications for collision-resistance.
Alagic, Bai, Poremba and Shi~\cite{alagic2023twosided} showed space-time trade-offs for two-sided permutation inversion---the task of inverting a random but invertible permutation on a random input. While the inverter does in fact receive access to an inverse oracle, it is crucial that such an oracle is \emph{punctured} at the challenge input; otherwise, the inversion task would become trivial. Although this work does not seem to have any immediate implications for the single-round sponge, it does suggest that permutation inversion remains hard in the presence of inverse oracles.

\paragraph{Concurrent work.} After completing our work, we became aware of a different manuscript on the single-round sponge by Majenz, Malavolta and Walter~\cite{cryptoeprint:2024/1140} which was prepared independently. They develop a generalization of Zhandry’s compressed oracle technique in the setting where the query algorithm has access to a permutation and its inverse, and they also obtain query lower bounds for both sponge inversion and Unruh's double-sided zero-search. While their framework also applies to permutation inversion with respect to arbitrary relations on (single)
input-output pairs, it does not yield tight query lower bounds as in our work.

\paragraph{Acknowledgments.} We thank Gorjan Alagic, Chen Bai, Luke Schaeffer and Kaiyan Shi for many useful discussions.
JC acknowledges funding from the Department of Defense.
AP is supported by the U.S. Department of Energy, Office of Science, National Quantum Information Science Research Centers, Co-design Center for Quantum Advantage (C2QA) under contract number DE-SC0012704.

\subsection{Organization}

In \Cref{sec:prelims}, we review the fundamentals of quantum computing and quantum search lower bounds. In \Cref{sec:subsetPairsSymGroup}, we review basic group theory---with a particular focus on the theory of Young subgroups---and state the key group-theoretic lemma for our symmetrization argument. We also introduce the concept of a \emph{zero pair} as well as a \emph{subset pair}, which constitute the main objects in the query problems we study. Moreover, we analyze many general combinatorial properties of these objects.

In \Cref{sec:queryLowerBounds}, we prove tight lower bounds for the \textsc{Double-Sided Zero-Search} problem using a worst-case to average-case reduction, as well as a direct proof based on the superposition oracle framework. We also define and show lower bounds for a non-uniform variant of this problem which we call \textsc{Non-uniform Double-Sided Search}---a problem which we naturally encounter later in \Cref{sec:spongeOneWay}. In the final section, we give a reduction from \textsc{Non-uniform Double-Sided Search} to the one-way game of the single-round sponge with invertible permutations.


\section{Preliminaries}
\label{sec:prelims}

Let us first introduce some basic notation and relevant background.

\paragraph{Notation.}
For $N\in \N$, we use $[N] = \{1,2,\dots,N\}$ to denote the set of integers up to $N$. The symmetric group on $[N]$ is denoted by $S_N$. 
In slight abuse of notation, we oftentimes identify elements $x \in [N]$ with bit strings $x \in \bit^n$ via their binary representation whenever $N=2^n$ and $n \in \N$. Similarly, we identify permutations $\pi \in S_N$ with permutations $\pi: \bit^{n} \rightarrow 
\bit^n$ over bit strings of length $n$.

We write $\negl(\cdot)$ to denote any \emph{negligible} function, which is a function $f$ such that, for every constant $c \in \mathbb{N}$, there exists an integer $N$ such that for all $n > N$, $f(n) < n^{-c}$.

\paragraph{Probability theory.} 

The notation $x \sim X$ describes that an element $x$ is drawn uniformly at random from the set $X$. Similarly, if $\algo D$ is a distribution, we let
$x \sim \algo D$ denote sampling $x$ according to $\algo D$.
We denote the expectation value of a random variable $X$ 
by $\mathbb{E}[X] = \sum_{x} x \Pr[ X = x] $. We make use of the following distribution.

\begin{definition}[Hypergeometric distribution]\label{defn:hypergeometricDistribution}
    A hypergeometric distribution describes the number of marked elements found by drawing $T$ many objects from a total of $N$, of which $K$ are marked. More formally, for $X \sim \text{Hypergeometric}(N,K,T)$, we have
    \begin{align*}
\Pr[X=k]=\frac{\binom{K}{k}\binom{N-K}{T-k}}{\binom{N}{T}}.
    \end{align*}
\end{definition}

We use the following tail bound for the hypergeometric distribution.
\begin{lemma}[Hoeffding's inequality, \cite{409cf137-dbb5-3eb1-8cfe-0743c3dc925f, Chvatal79hypergeometric}]\label{lem:hoeffding} Let
$X \sim \text{Hypergeometric}(N,K,T)$ and define $p = K/N$. Then, for any $t \in \mathbb{R}$ with $0 \leq t < 1-p$, it holds that
$$
\Pr\left[X \geq (p+t) \cdot T \right] \leq \exp\big(- T \cdot D_{\mathrm{KL}}(p +t || p)  \big) \, ,
$$
where $D_{\mathrm{KL}}(q || p)$ denotes the Kullback-Leibler divergence with
$$
D_{\mathrm{KL}}(q || p) = q \cdot \ln\left(\frac{q}{p} \right) + (1-q) \cdot \ln\left(\frac{1-q}{1-p} \right).
$$
\end{lemma}

\paragraph{Quantum computing.} A finite-dimensional complex Hilbert space is denoted by $\algo H$, and we use subscripts to distinguish between different systems (or registers); for example, we let $\algo H_{A}$ be the Hilbert space corresponding to a system $A$. 
The tensor product of two Hilbert spaces $\algo H_{A}$ and $\algo H_{B}$ is another Hilbert space which we denote by $\algo H_{AB} = \algo H_{A} \otimes \algo H_{B}$.  We let $\algo L(\algo H)$ denote the set of linear operators over $\algo H$. A quantum system over the $2$-dimensional Hilbert space $\algo H = \mathbb{C}^2$ is called a \emph{qubit}. For $n \in \mathbb{N}$, we refer to quantum registers over the Hilbert space $\algo H = \big(\mathbb{C}^2\big)^{\otimes n}$ as $n$-qubit states. We use the word \emph{quantum state} to refer to both pure states (unit vectors $\ket{\psi} \in \algo H$) and density matrices $\rho \in \algo D(\algo H)$, where we use the notation $\algo D(\algo H)$ to refer to the space of positive semidefinite linear operators of unit trace acting on $\algo H$. 

The \emph{trace distance} between two density matrices $\rho,\sigma \in \mathcal{D}(\mathcal{H)}$ is given by
$$
\mathsf{TD}(\rho,\sigma) = \frac{1}{2} \Tr\left[ \sqrt{ (\rho - \sigma)^\dag (\rho - \sigma)}\right].
$$
Note that the trace distance between two pure states $\ket{\psi},\ket{\phi} \in \big(\mathbb{C}^2\big)^{\otimes n}$ satisfies
$$
\mathsf{TD}(\proj{\psi},\proj{\phi}) \leq \big\| \hspace{-0.8mm}\ket{\psi} - \ket{\phi} \hspace{-0.8mm}\big\| \, ,
$$
where $\| \cdot \|$ denotes the Euclidean distance over the vector space $\big(\mathbb{C}^2\big)^{\otimes n}$.

When the dimensions are clear from context, we use $\Id$ to denote the identity matrix. A \emph{unitary} $U: \algo L (\algo H_{A}) \to \algo L(\algo H_{A})$ is a linear operator such that $U^\dagger U = U U^\dagger = \Id_A$. A \emph{projector} $\Pi$ is a Hermitian operator such that $\Pi^2 = \Pi$. Oftentimes, we use the shorthand notation $\bar{\Pi} = \Id - \Pi$.

A quantum algorithm is a uniform family of quantum circuits $\{\algo A_\lambda\}_{\lambda \in \mathbb{N}}$, where each circuit $\algo A_\lambda$ is described by a sequence of unitary gates and measurements; moreover, for each $\lambda \in \mathbb{N}$, there exists a deterministic Turing machine that, on input $1^\lambda$, outputs a circuit description of $\algo A_\lambda$.
We say that a quantum algorithm $\mathcal{A}$ has oracle access to a classical function $f: \{0,1 \}^{n} \rightarrow \{0,1 \}^m$, denoted by $\mathcal{A}^f$, if $\mathcal{A}$ is allowed to use a unitary gate $O^f$ at unit cost in time. The unitary $O^f$ acts as follows on the computational basis states of a Hilbert space $\mathcal{H}_X \otimes \mathcal{H}_Y$ of $n+m$ qubits:
$$
O^f: \quad
\ket{x}_X \otimes \ket{y}_Y \longrightarrow \ket{x}_X \otimes \ket{y \oplus f(x)}_Y,
$$
where the operation $\oplus$ denotes bit-wise addition modulo $2$. Oracles with quantum query-access have been studied extensively, for example in the context of quantum complexity theory~\cite{Bennett_1997}, as well as in cryptography~\cite{10.1007/978-3-642-25385-0_3,cryptoeprint:2018/904,cryptography4010010,}.

\subsection{Lower bounds for unstructured search}

The quantum query complexity of worst-case \textsc{Unstructured Search} is one of the most fundamental results in quantum computing~\cite{Grover96algorithm,Bennett_1997,Boyer_1998}. In this problem, a quantum algorithm is given query access to a function $f:[N] \rightarrow \bit$ and is tasked with finding an $x \in [N]$ such that $f(x)=1$. We will consider this problem subject to the promise that there are exactly $K$ preimages of $1$, i.e., $|f^{-1}(1)|=K$. Note also that only ever consider the worst-case complexity of this problem, meaning the success probability is only over the internal randomness of the algorithm.

The exact query complexity of this problem is well-understood \cite{Grover96algorithm,Zalka99optimal,10.5555/2011791.2011803}. We use the following lower bound, which is a direct corollary of the work of Dohotaru and H\o yer.

\begin{theorem}[\cite{10.5555/2011791.2011803}, Theorem 8] Any algorithm which solves \textsc{Unstructured Search} with $K$ out of $N$ marked elements with worst-case success probability $\epsilon >0$ requires at least
$$
T \geq \frac{\sqrt{N/K}}{2\sqrt{2}} \left(1 + \sqrt{\epsilon} - \sqrt{1 -\epsilon} - \frac{2}{\sqrt{N/K}} \right)
$$
many queries.
    \label{thm:unstructuredSearchHard}
\end{theorem}

We can rearrange the expression in \Cref{thm:unstructuredSearchHard} to obtain a bound on the worst-case success probability of a $T$-query quantum algorithm.

\begin{corollary}
    Any algorithm which solves \textsc{Unstructured Search} with $K$ marked elements out of $N$ with worst-case success probability $\epsilon >0$, satisfies the inequality
$$
\epsilon \leq \frac{8(T+1)^2K}{N}.
$$
    \label{cor:unstructuredSearchHardErrorVersion}
\end{corollary}
\begin{proof}
    We have \begin{align*}
        T \geq& \frac{\sqrt{N/K}}{2\sqrt{2}} \left(1 + \sqrt{\epsilon} - \sqrt{1 -\epsilon} - \frac{2}{\sqrt{N/K}} \right) & \text{(\Cref{thm:unstructuredSearchHard})} \\
         \geq& \frac{\sqrt{N/K}}{2\sqrt{2}} \left(\sqrt{\epsilon} - \frac{2}{\sqrt{N/K}} \right).
    \end{align*}
    This implies that $T+1 \geq \sqrt{N\epsilon/K}/\sqrt{8}$, which yields the desired inequality.
\end{proof}
This bound is no longer exact, but will suffice for our purposes.

\section{Subset Pairs and the Symmetric Group}
\label{sec:subsetPairsSymGroup}

In this section, we recall some basic facts from the theory of Young subgroups. Due to the nature of the query problems in \Cref{sec:queryLowerBounds}, we oftentimes encounter subgroups of the symmetric group $S_N$ consisting of permutations that fix certain subsets of $[N]$; for example, zero pairs in the case of \textsc{Double-Sided Zero-Search} in \Cref{prob:doubleSidedZero}.
This forces us to analyze certain subgroups of $S_N$---called Young subgroups---which are direct products of subgroups on fewer elements that form a partition of $[N]$. 

We use the machinery of Young subgroups to prove a key lemma which shows our symmetrization procedure is sound, and to offer an algebraic characterization of the input/output pairs of a permutation satisfying certain constraints.

\subsection{Group Theory}

We first define relevant notions in group theory, and recall known results about the symmetric group $S_N$ on $N$ elements. For a more complete overview of the subject, we refer the reader to the work of James \cite{James_1984}.


Let $S_N$ denote the symmetric group consisting of permutations which act on the set $[N]:=\{1, ...,N\}$. For a subset $A \subset [N]$, let $S_A$ denote the maximal subgroup of $S_N$ which fixes every element in the complement of $A$, i.e. $[N] \setminus A$. Now let $A_1,...,A_l$ be a partition of $[N]$ such that the disjoint union satisfies $\bigsqcup_{i \in [\ell]} A_i =[N]$.

\begin{definition}[Young subgroup]
    A subgroup $H$ of the symmetric group $S_N$ is a Young subgroup if it can be expressed as $H=S_{A_1} \times ... \times S_{A_l}$, where $\times$ denotes the internal direct product and the collection of subsets $\{A_i\}_{i \in [\ell]}$ forms a partition of $[N]$. \label{definition:youngSubgroup}
\end{definition}

The concept of a double coset, which we review below, will also be relevant.

\begin{definition}[Double cosets] Let $H,K$ be subgroups of a group $G$. The double cosets of $G$ under $(H, K)$, denoted $H \diagdown G \diagup K$, are the sets of elements which are invariant under left multiplication by $H$ and right multiplication by $K$. In particular, 
$$
H \diagdown G \diagup K = \big\{\{hxk\, : \,h \in H, k \in K\}\, : \,x \in G\big\}.$$ \label{definition:doubleCoset}
\end{definition}
\vspace{-2mm}
It is well-known that $G$ is the disjoint union of its double cosets for any subgroups $H,K \leq G$. We focus on the double cosets of the symmetric group $S_N$, specifically those generated by Young subgroups. These subgroups admit the following characterization, adapted from Jones \cite{Jones1996DoubleCosets} and James \cite{James_1984}.

\begin{theorem}[\cite{Jones1996DoubleCosets}, Theorem 2.2]
    Let $H, K$ be Young subgroups of $S_N$, with corresponding partitions $A_1,...,A_l$ (for $H$) and $B_1,...,B_m$ (for $K$). Let $\pi \in S_N$ and $C=H\pi K$ be the corresponding double coset. Any other permutation $\pi' \in S_N$ is in $C$ if and only if for all $i \leq l, j \leq m$ we have $|A_i \cap \pi' B_j| = |A_i \cap \pi B_j|$.  \label{theorem:characterizationYoungDoubleCosets}
\end{theorem}

Intuitively, the above characterization says that the double cosets defined by Young subgroups $(H, K)$ correspond to sets of permutations which look the same if one only considers how they distribute the elements of each $B_j$ among the different $A_i$'s. Two permutations $\pi, \pi'$ are in the same double coset if and only if for every $A_i, B_j$ both $\pi$ and $\pi'$ send the same number of elements from $B_j$ to $A_i$. This characterization will prove useful when showing the random self-reducibility of two-sided search problems. The following lemma and proof is a key component of this reduction, and is based heavily on the approach of Wildon \cite{wildonSymGroupRep}, Theorem 4.4.

\begin{lemma} For any subgroups $H, K$ of a (finite) group $G$ with $x, g \in G$ both in the same $(H, K)$ double coset, there are exactly $|x^{-1}Hx \cap K|$ ways of choosing $h \in H$ and $k \in K$ such that $g=h x k$. \label{lemma:cardinalityDoubleCosets}
\end{lemma}

\begin{proof}   

Fix some $x \in G$, and consider any $g\in G$ in the same double coset as $x$. Then, there exist elements $h \in H$, $k \in K$ such that \begin{align}
    g =& h x k \nonumber \\
    =& x (x^{-1} h x) k. \nonumber
\end{align}
Any alternative $k' \in K$ can be written as $k'=u_kk$ for $u_k \in K$, and an alternative $h' \in H$ can be written as $x^{-1}h'x=(x^{-1}hx)u_h$ for $u_h \in x^{-1}H x$. These new elements satisfy \begin{align}
    g' =& h' x k' \nonumber \\
    =& x (x^{-1} h x) u_hu_kk \nonumber
\end{align}
It is apparent that $g'=g$ holds if and only if $u_h = u_k^{-1}$. This implies that 
$$
u_k, u_h \in K \cap x^{-1}Hx.
$$
Thus, there are $|x^{-1}Hx \cap K|$ alternative ways to choose a satisfying $h \in H$ and $k \in K$.

\end{proof}

\subsection{Subset Pairs}
In this section, we introduce the notion of a subset pair, which will be the main object of study. Specifically, we are interested in the complexity of finding input/output pairs of a permutation which lie within certain pairs of subsets.

\paragraph{Zero pairs.} Before giving the general definition of a subset pair, we first define a zero pair as a specific and relevant example of a subset pair.
\begin{example}[Zero pair]\label{example:zeroPair}
Let $Z \subset [2^{2n}]$ be the subset consisting of all $z \in [2^{2n}]$ with binary decomposition $(z_1,z_2,\dots,z_{2n}) \in \bit^{2n}$ such that 
\begin{align*}
z_{n+1} = z_{n+2} = \hdots = z_{2n} = 0.
\end{align*}
\end{example}
Then, $Z$ is a subset of size $2^n$ and, when identifying $S_{2^{2n}}$ with permutations over $2n$-bit strings, $(x||0^n)$ and $(y||0^n)$ form a $Z$-pair for a permutation $\pi: \bit^{2n} \rightarrow \bit^{2n}$, if $\pi(x||0^n) = y ||0^n$. For brevity, we sometimes call $(x,y)$ a zero pair of the permutation $\pi$. We show the following concerning zero pairs in random permutations.

\begin{lemma} 
    A random permutation $\pi \sim S_{2^{2n}}$ contains at least one zero pair with probability $$\underset{\pi \sim S_{2^{2n}}}{\Pr}[|Z_\pi| \geq 1] = 1-1/e+o(1).$$
    Here, $|Z_\pi|$ denotes the random variable for the number of zero pairs of the permutation $\pi$.
    \label{fact:Z-pairs-existence-probability}
\end{lemma}
\begin{proof}
Let $N=2^{2n}$ and let $Z \subset [N]$ be the subset from \Cref{example:zeroPair}, where $|Z| = \sqrt{N}$.
It suffices the estimate the probability that a random permutation $\pi \sim S_{N}$ has no zero pairs. The total number of permutations over $[N]$ which have no zero pairs is
$$
\frac{(N-\sqrt{N})!}{(N-2\sqrt{N})!} \cdot (N-\sqrt{N})! \,\,.
$$
This is because we need to match all $\sqrt{N}$ elements in the subset $Z$ exclusively with elements outside of $Z$. Therefore, we get
\begin{align*}
\underset{\pi \sim S_N}{\Pr}[|Z_\pi| \geq 1] &= 1- \underset{\pi \sim S_N}{\Pr}[|Z_\pi| =0]\\
&= 1 - \frac{(N-\sqrt{N})! \cdot (N-\sqrt{N})!}{N! \cdot (N-2\sqrt{N})!}\\
&= 1- \binom{N - \sqrt{N}}{\sqrt{N}}/\binom{N}{\sqrt{N}}.
\end{align*}
Therefore, the claim follows from the fact that
$\binom{N - \sqrt{N}}{\sqrt{N}}/\binom{N}{\sqrt{N}}$ is a monotonically increasing sequence which approaches the limit $1/e$ as $N \rightarrow  \infty$.  
\end{proof}

We will use the notation $G_Z$ to denote the subgroup $G_Z \leq S_{2^{2n}}$ of permutations which preserve membership in $Z$.

\begin{definition}[Subgroup $G_Z$]
    A permutation $\pi$ is a member of the subgroup $G_Z \leq S_{2^{2n}}$ if and only if $\pi$ fixes the subset $Z$ such that $\pi(Z)=Z$. In other words, for every $x \in \bit^n$, there exists $y \in \bit^n$ such that $\pi(x||0^n)=y||0^n$.
    \label{defn:zeroPairSubgroup}
\end{definition}
Note that the subgroup $G_Z$ decomposes as $G_Z \cong S_{Z} \times S_{[N] \setminus Z}$. 
The algebraic structure of the double cosets induced by $G_Z \times G_Z$ will be central in proving the hardness of \textsc{Double-Sided Zero-Search} in \Cref{prob:doubleSidedZero}.

\paragraph{General statements.} We now turn to defining a subset pair in full generality, and prove a key ``symmetrization lemma'' using the theory of Young subgroups.

\begin{definition}[Subset pair]\label{def:subset-pair} Let $N \in \N$, let $X_1, X_2 \subseteq [N]$ be subsets and $X=(X_1,X_2)$. We say that a pair $(i,j)$ is an $X$-pair for a permutation $\pi \in S_N$ if $i\in X_1, j \in X_2$ and $\pi(i) = j$.
We denote the set of all $X$-pairs of a permutation $\pi \in S_N$ by
$$
X_\pi = \{ (i,j) \, | \,  i\in X_1, j \in X_2 \, \text{ and } \, \pi(i) = j\}.
$$
\end{definition}

\begin{remark} Whenever $X= (X_1,X_2)$ is a pair of identical subsets with $X_1=X_2=Z$, for some $Z \subseteq [N]$, we sometimes refer to the pairs $(i,j)$ in \Cref{def:subset-pair} as $Z$-pairs instead.
This is the case for the zero pair subsets
from \Cref{example:zeroPair}.
\end{remark}

For an $X=(X_1,X_2)$, define the subgroups $G_1 \leq S_N$ and $G_2 \leq S_N$ with
\begin{align*}
    G_1 :=& \{ \pi \in S_{N} \, : \, \pi(X_1) = X_1\} \\
    G_2 :=& \{ \pi \in S_{N} \, : \, \pi(X_2) = X_2\}.
\end{align*}
In other words, $G_i$ consists of all permutations in $S_N$ that fix $X_i$. Moreover, it follows from Definition \ref{definition:doubleCoset} that $G_1$ and $G_2$ are Young subgroups of $S_N$.

\begin{fact} The Young subgroups $G_1 \leq S_N$ and $G_2 \leq S_N$ decompose as
$$
G_i \cong S_{X_i} \times S_{[N] \setminus X_i}, \quad\quad \text{ for } i \in \{1,2\}.
$$
\end{fact}
We introduce the following notation for the double cosets of $S_N$ under $(G_1,G_2)$. 
\begin{definition}    \label{defn:youngDoubleCosetsNotation}
    For $\kappa \in \mathbb{N} \cup \{0\}$, and $X=(X_1, X_2)$ with $X_1, X_2 \subset [N]$, define the subset
    \begin{align*}
         S_{N}^{\kappa, X} := \{\pi \in S_N \, : \, \text{$\pi$ has exactly $\kappa$ many $X$-pairs}\}.   
    \end{align*}
\end{definition} 
We oftentimes omit the superscript $X$ when it is clear from context. Note that $S_N = \bigsqcup_{\kappa=0}^{|X_1|} S_N^\kappa$, i.e., $S_N$ is the disjoint union over all the $S_N^\kappa$. Note, however, that whenever we have $\kappa > |X_1|$ or $\kappa > |X_2|$, it must be that $S_N^\kappa=\emptyset$. The following corollary comes from the characterization in Theorem \ref{theorem:characterizationYoungDoubleCosets}, since $G_1$ and $G_2$ are both Young subgroups.
\begin{corollary}\label{corrolary:zeroSearchYoungSubgroups}
The sets $\{S_N^\kappa\}$ are the double cosets of $S_N$ under $(G_1, G_2)$; formally, 
  \begin{align*}
 G_1 \diagdown S_N \diagup G_2 = \{S_N^\kappa \,|\, \kappa \in \{0,1,\dots,|X_1|\}\}.
  \end{align*}
\end{corollary}
\begin{proof}
From \Cref{theorem:characterizationYoungDoubleCosets}, a permutation $\varphi$'s double coset is determined by the data $(|\varphi(X_1) \cap X_2|, |\varphi(X_1) \cap ([N] \setminus X_2)|, |\varphi([N] \setminus X_1) \cap  X_2|, |\varphi([N] \setminus X_1) \cap ([N] \setminus X_2)|)$. Note that once the number of $X$-pairs (equivalently, $|\varphi(X_1) \cap X_2|$) is fixed the remaining quantities are determined. In particular, we have \begin{align*}
    |\varphi(X_1) \cap ([N] \setminus X_2)| =& |X_1| - |\varphi(X_1) \cap X_2|, \\
    |\varphi([N] \setminus X_1) \cap  X_2| =& |X_2| - |\varphi(X_1) \cap X_2|, \\
    |\varphi([N] \setminus X_1) \cap ([N] \setminus X_2)| =& N-|X_1|-|X_2| + |\varphi(X_1) \cap X_2|,
\end{align*}
and hence the number of $X$ pairs uniquely determines $\varphi$'s double coset.
\end{proof}

We are now ready to prove the main tool for analyzing random double-sided search problems, such as \textsc{Double-Sided Zero-Search} in \Cref{prob:doubleSidedZero}.

\begin{lemma}[Symmetrization lemma]
Consider any fixed permutation $\pi \in S_N^\kappa$, for some $\kappa \in \mathbb{N} \cup \{0\}$, and suppose that $\omega \sim G_1, \sigma \sim G_2$ are sampled independently and uniformly at random. Then, $\omega \circ \pi \circ \sigma$ is a uniformly random element of $S_N^\kappa$. \label{lemma:rerandomizationZeroSearch}
\end{lemma}
\begin{proof}
    From Lemma \ref{lemma:cardinalityDoubleCosets}, it follows that for any fixed $\pi \in S_N^\kappa$ and some $\varphi \in S_N^\kappa$ the number of $\omega, \sigma$ pairs with $\omega \in G_1, \sigma \in G_2$ such that $\varphi=\omega \circ \pi \circ \sigma$ is $|\pi G_1\pi^{-1} \cap G_2|$, which notably is independent of $\varphi$. In particular, we have \begin{align*}
        \Pr_{\omega \sim G_1, \sigma \sim G_2}[\omega \circ \pi \circ \sigma =\varphi] = \frac{|\pi G_1\pi^{-1} \cap G_2|}{\sum_{\varphi' \in S_N^k} |\pi G_1\pi^{-1} \cap G_2|} = \frac{1}{|S_N^\kappa|}.
    \end{align*}
The probability is the same for all elements $\varphi \in S_N^\kappa$, and so the element $\omega \circ \pi \circ \sigma$ is uniformly random among elements of $S_N^\kappa$.
\end{proof}

\subsection{Combinatorics of Subset Pairs}
\label{sec:combinatoricsSubsetPairs}

In this section, we work out the combinatorics behind the expected number of subset pairs (on average over the choice of permutation), as well derive strong tail bounds in the case where the expected number is small. We show that the number of subset pairs decays exponentially for the uniform distribution (see \Cref{thm:X-pairs-uniform-tail-bound}), as well as for a particular choice of non-uniform distribution (see \Cref{thm:X-pairs-nonuniform-tail-bound}). 

\paragraph{Uniform case.}
We begin by considering the case of uniformly random $\pi \sim S_N$ and derive the expected number of $X$-pairs as well as tail bounds. We show the following:

\begin{theorem}
Let $N \in \N$ and let $X_1, X_2 \subseteq [N]$ be subsets with $X=(X_1,X_2)$. Then, on average over the uniform choice of $\pi \sim S_N$, the expected number of $X$-pairs equals
$$
\underset{\pi \sim S_N}{\E}[|X_\pi|]=\frac{|X_1||X_2|}{N}.
$$
    \label{thm:X-pairs-uniform}
\end{theorem}

\begin{proof}
    For a uniform random $\pi \sim S_N$, we let $I_x$ denote the indicator variable which equals $1$, if $\pi(x) \in X_1$, and $0$ otherwise. Thus, we have
    \begin{align*}
        \E[I_x] =& \frac{|X_2|}{N}, &
        |X_\pi| =& \sum_{x \in X_1} I_x \, . 
    \end{align*}  From the above, it then follows that
    \begin{align*}
        \underset{\pi \sim S_N}{\E}[|X_\pi|] =& \sum_{x \in X_1} \E[I_x] & \text{(Linearity of Expectation)} \\
        =& \frac{|X_1||X_2|}{N}.
    \end{align*}
\end{proof}

A simple corollary of this bound is an alternative characterization of the total number of subset pairs over all permutations.

\begin{corollary}
Let $N \in \N$ and fix arbitrary subsets $X_1, X_2 \subseteq [N]$ with $X=(X_1,X_2)$. Then, there exist precisely $$
    \sum_{\pi \in S_N}|X_\pi|=\frac{|X_1||X_2|}{N} \cdot N!$$ 
    many $X$-pairs over all permutations $\pi \in S_N$.
    \label{corollary:total-X-pairs}
\end{corollary}
    
\begin{proof}
    This is easily verified as follows: \begin{align*}
        \sum_{\pi \in S_N}|X_\pi| =& \underset{\pi \sim S_N}{\mathbb{E}}[|X_\pi|]\cdot|S_N| \\
        =& \frac{|X_1||X_2|}{N} \cdot N! \, . &\text{(\Cref{thm:X-pairs-uniform})}
    \end{align*}
\end{proof}

We now state and prove tail bounds, which will be sufficiently tight for the case where the expected number of subset pairs is small.

\begin{theorem}
   Let $N \in \N$ and $X_1, X_2 \subseteq [N]$ be subsets with $X=(X_1,X_2)$. Then, for any real number $u \geq 6 \, \mathbb{E}_{\sigma \sim S_N}[|X_\sigma|]$, it holds that \begin{align*}
    \underset{\pi \sim S_N}{\Pr}\Big[|X_\pi|\geq \underset{\sigma \sim S_N}{\mathbb{E}}[|X_\sigma|] + u \Big] \leq \exp\left(-\frac{3}{4}u\right).
    \end{align*}
    \label{thm:X-pairs-uniform-tail-bound}
\end{theorem}
\begin{proof}
   The number of permutations having $\kappa$ many $X$-pairs for $0 \leq \kappa\leq |X_1|, |X_2|$ is hypergeometrically distributed (see \Cref{defn:hypergeometricDistribution}). In other words,
   \begin{align*}
        \Pr_{\pi \sim S_N}[|X_{\pi}|=\kappa] =& \frac{\binom{|X_2|}{\kappa}\binom{N-|X_2|}{|X_1|-\kappa}}{\binom{N}{|X_1|}}.
    \end{align*}
    This is because each image of $X_1$ can be assigned without replacement by a random permutation to one of $N$ images. Next, we want to count how many images there are in $X_2$, which we interpret as a ``success''. Let us consider the hypergeometric distribution $\text{Hypergeometric}(N,K,T)$ with the following parameters: \begin{align*}
        N :=& N & \text{(Total Objects)} \\
        K :=& |X_2| & \text{(Success Objects)} \\
        T :=& |X_1| & \text{(Trials)}
    \end{align*}
    Let $p=|X_2|/N$ denote the probability that the first draw is a success. Applying Hoeffding's inequality (see~\Cref{lem:hoeffding}), we get \begin{align*}
        \Pr_{\pi \sim S_N}\Big[|X_\pi|\geq (p+t)|X_1| \Big] \leq& \exp\left(-|X_1|D_{\mathrm{KL}}(p+t||p)\right) \\
        \leq& \exp\left(-\frac{3}{4}|X_1|t\right) & \text{(From \Cref{lem:KLDivergenceBound})}.
    \end{align*}
    By substituting $u=|X_1|t$ and applying \Cref{thm:X-pairs-uniform}, we obtain
    \begin{align*}
        \underset{\pi \sim S_N}{\Pr}\Big[|X_\pi|\geq \underset{\sigma \sim S_N}{\mathbb{E}}[|X_\sigma|] + u \Big] \leq \exp\left(-\frac{3}{4}u\right) && \text{(When $u\geq 6\underset{\sigma \sim S_N}{\mathbb{E}}[|X_\sigma|]$)}.
    \end{align*}
\end{proof}

We use the following technical lemma.

\begin{lemma}
    Let $p, t \in (0,1)$ such that $t \geq 6p$ and $p+t<1$. Then, we can lower bound the Kullback-Leibler divergence as follows:
    \begin{align*}
        D_{\mathrm{KL}}(p +t || p) > 3t/4.
    \end{align*}
    \label{lem:KLDivergenceBound}
\end{lemma}
\begin{proof}
    Recall that the Kullback-Leibler divergence $D_{\mathrm{KL}}(p +t || p)$ is defined as
    
    \begin{align*}
        D_{\mathrm{KL}}(p+t||p) =& \underbrace{(p+t)\ln\left(\frac{p+t}{p}\right)}_{=D_1}
        + \underbrace{(1-p-t)\ln\left(\frac{1-p-t}{1-p}\right)}_{=D_2}.
    \end{align*}
    We now bound $D_1$ and $D_2$ separately. 
    Beginning with the first term, we get \begin{align*}
        D_1 > t \ln\left(\frac{p+t}{p}\right)> t \ln\left(\frac{t}{p}\right) \geq t \ln(6) \, ,
    \end{align*}
    where we used that $t\geq 6p$.
    We now continue with the second term. First, we consider the following first-order derivatives with respect to $t$: \begin{align*}
        \frac{\partial}{\partial t} \left(D_2\right)=& -\ln\left(\frac{p+t-1}{p-1}\right)-1 \\
        \frac{\partial}{\partial t} \left(D_2\right) \big |_{t=0} =& -1 \, .
    \end{align*}
 Calculating the second-order derivatives, we also find that
 \begin{align*}
        \frac{\partial^2}{\partial t^2} \left(D_2\right)=&-\frac{1}{p+t-1} \\
        \frac{\partial^2}{\partial t^2} \left(D_2\right) >& \, 0 \, . & \text{(From $p+t<1$)}.
    \end{align*}
    It follows that $D_2$ is concave upwards in the relevant regime, and the derivative at $t=0$ is $-1$. We further have $D_2|_{t=0}=0$, so when $p+t<1$ and $t \geq 0$ we have $ D_2 > -t$.
    Combining these results, we obtain the desired inequality:$$
        D(p+t||p) > t \ln (6) - t > 3t/4.
 $$
\end{proof}

\paragraph{Non-uniform case.}

In this section, we work out the combinatorics for the number of subset pairs with respect to a particular \emph{non-uniform} distribution which assigns more weight to permutations that have a large amount of $X$-pairs.

\begin{definition}
Let $N \in \N$. The non-uniform distribution $\algo D_X$ over permutations in $S_N$ is parameterized by a pair of subsets $X=(X_1,X_2)$, and is defined as follows:
    $$
    \underset{\Phi \sim \algo D_X}{\Pr}[\Phi = \varphi] = \frac{|X_\varphi|}{\sum_{\, \sigma \in S_N} |X_\sigma|} \,\,, \,\quad \text{ for } \, \varphi \in S_N.$$
    \label{defn:nonUniformDist}
\end{definition}

When $X$ is clear from the context, we sometimes denote $\algo D_X$ by $\algo D$. We begin by bounding the average number of $X$-pairs.

\begin{theorem}
    Let $N \in \N$ and $X_1, X_2 \subset [N]$ be non-empty subsets. Let $|X_\pi|$ denote the number of $X$-pairs with respect to $\pi \in S_N$. Then, the average number of $X$-pairs over the choice of permutation from $\pi\sim\algo D_X$ satisfies
    $$
    1  \, \leq \, \underset{\pi \sim \algo D_X}{\E}[|X_\pi|] \,\leq \, 1 + \frac{|X_1||X_2|}{N}.
    $$
    \label{thm:X-pairs-nonuniform}
\end{theorem}

\begin{proof}
Recall that the subsets $\{S_N^\kappa\}$ form a partition $S_N = \bigsqcup_{\kappa=0}^{|X_1|} S_N^\kappa$, i.e., $S_N$ is the disjoint union over all the $S_N^\kappa$, we have $N! = \sum_{\kappa=0}^{|X_1|} |S_N^\kappa|$, and thus
$$
\underset{\sigma \sim S_N}{\E}[|X_\sigma|] = \sum_{\kappa = 0}^{|X_1|} \kappa \cdot \Pr_{\sigma \sim S_N}[\sigma \in S_N^\kappa] = \sum_{\kappa = 0}^{|X_1|} \kappa \cdot \frac{|S_N^\kappa|}{N!}.
$$
Therefore, we obtain the following identity,
\begin{align}\label{eq:identity-sum-Z}
\sum_{\, \sigma \in S_N} |X_\sigma| = \sum_{\kappa=0}^{|X_1|} \sum_{\sigma \in S_N^\kappa} |X_\sigma| = \sum_{\kappa=0}^{|X_1|} \kappa  \cdot |S_N^\kappa| = N! \cdot \underset{\sigma \sim S_N}{\E}[|X_\sigma|].
\end{align}
We can now write the average number of $X$-pairs as follows:
\begin{align*}
\underset{\pi \sim \algo D_X}{\E}[|X_\pi|]
&= \sum_{\kappa = 0}^{|X_1|} \kappa \cdot \Pr_{\pi \sim \algo D_X}[\pi \in S_N^\kappa]\\
&= \sum_{\kappa = 0}^{|X_1|} \kappa \cdot \sum_{\pi \in S_N^\kappa}\Pr_{\Pi \sim \algo D_X}[\Pi=\pi]\\
&= \sum_{\kappa = 0}^{|X_1|} \kappa \cdot \frac{\sum_{\pi \in S_N^\kappa} |X_\pi|}{\sum_{\sigma \in S_N} |X_\sigma|}\\
&= \sum_{\kappa = 0}^{|X_1|} \kappa \cdot \frac{\kappa \cdot |S_N^\kappa|}{\sum_{\sigma \in S_N} |X_\sigma|}\\
&= \left(\underset{\sigma \sim S_N}{\E}[|X_\sigma|]\right)^{-1} \sum_{\kappa = 0}^{|X_1|} \kappa^2 \cdot \frac{|S_N^\kappa|}{N!} & \text{(By \Cref{{eq:identity-sum-Z}})}\\
&=\frac{\underset{\sigma \sim S_N}{\E}[|X_\sigma|^2]}{\underset{\sigma \sim S_N}{\E}[|X_\sigma|]}. \numberthis \label{eqn:second-moment-identity}
\end{align*}
In other words, the average number of zero pairs in the non-uniform case is the ratio between the second moment and the first moment in the uniform case. We also have \begin{align*}
    \frac{|X_1||X_2|}{N} \leq \underset{\sigma \sim S_N}{\E}[|X_\sigma|^2] \leq & \frac{|X_1||X_2|}{N} + \frac{|X_1|^2|X_2|^2}{N^2} & \text{(\Cref{lem:xPairsSecondMoment})} \\
    \underset{\sigma \sim S_N}{\E}[|X_\sigma|] =& \frac{|X_1||X_2|}{N}. & \text{(\Cref{thm:X-pairs-uniform})}
\end{align*}
Putting everything together, we find that \begin{align*}
    1 \leq \underset{\pi \sim \algo D}{\E}[|X_\pi|] \leq 1 + \frac{|X_1||X_2|}{N}.
\end{align*}
\end{proof}

Recall that the previous lemma relied on the following fact:

\begin{lemma}
    Let $N \in \N$ and $X_1, X_2 \subset [N]$ be non-empty subsets. Let $|X_\pi|$ denote the random variable for the number of $X$-pairs for $\pi \in S_N$. Then, the second moment satisfies 
    $$
    \left(\frac{|X_1||X_2|}{N}\right) \leq \underset{\pi \sim S_N}{\E}[|X_\pi|^2] \leq \left(\frac{|X_1||X_2|}{N} + \frac{|X_1|^2|X_2|^2}{N^2}\right).$$
    \label{lem:xPairsSecondMoment}
\end{lemma}

\begin{proof}
    For a uniform random $\pi \sim S_N$ and some $x \in [N]$, we let $I_x$ be the indicator variable which equals $1$, if $\pi(x) \in X_2$, and $0$ otherwise. Then, we find that
    \begin{align*}
        \E[I_x] =& \frac{|X_2|}{N}, &
        |X_\pi|^2 =& \left(\sum_{x \in X_1} I_x \right)^2
    \end{align*} Therefore, we can express the second moment as follows:
    \begin{align*}
        \underset{\pi \sim S_N}{\E}[|X_\pi|^2] =& \sum_{x, y \in X} \E[I_xI_y] & \text{(Linearity of Expectation)} \\
        =& \sum_{x \in X_1} \E[I_x] + \sum_{x \neq y \in X_1} \E[I_xI_y].
    \end{align*}
    Note that the the first term is equal to  $$
    \underset{\sigma \sim S_N}{\E}[|X_\sigma|] = \frac{|X_1||X_2|}{N}.
    $$
    The second term is obviously positive, hence the lower bound. The second term has less than $|X_1|^2$ summands, and each summand can be bounded above via
    $$
    \frac{|X_2|(|X_2|-1)}{N(N-1)} < \frac{|X_2|^2}{N^2}.$$
    Multiplying this through, we can upper bound the term as $\frac{|X_1|^2|X_2|^2}{N^2}$. This gives the desired inequality.
\end{proof}

We can now give a tail bound for the number of $X$-pairs in the non-uniform case. Our bound is tailored to the case where $|X_1| \cdot |X_2|=N$.

\begin{theorem} Let $N \in \N$ and let
$X=(X_1,X_2)$ be a pair of subsets $X_1, X_2 \subseteq [N]$ such that $|X_1|\cdot |X_2|=N$. Then, for any real number $u \geq 6\underset{\sigma \sim \algo D_X}{\E}[|X_\sigma|]$, it holds that
\begin{align*}
        \underset{\pi \sim \algo D_X}{\Pr} \Big[|X_\pi|\geq \underset{\sigma \sim \algo D_X}{\E}[|X_\sigma|] + u \Big] \leq 3\exp\left(-4u/9\right).
    \end{align*}
    \label{thm:X-pairs-nonuniform-tail-bound}
\end{theorem}
\begin{proof}
    
    Note that because $|X_1|\cdot |X_2|=N$ we can write $\algo D_X$ as 
    \begin{align*}
        \underset{\Phi \sim \algo D_X}{\Pr}[\Phi = \varphi] =& \frac{|X_\varphi|}{\sum_{\, \sigma \in S_N} |X_\sigma|} \\
        =& \frac{|X_\varphi|}{N!}&\text{(\Cref{corollary:total-X-pairs})}
    \end{align*}
   Using the above equation, we can now write the probability distribution on $X$-pairs as follows:\begin{align*}
        \Pr_{\pi \sim \algo D} [|X_\pi|= \kappa] =& \kappa \cdot \Pr_{\sigma \sim S_N} [|X_\sigma|=\kappa],
    \end{align*}
for any $\kappa \in \{0,1,\dots,|X_1|\}$. In particular, for $\kappa>6$, we have the bound\begin{align*}
        \Pr_{\pi \sim S_N} [|X_\pi|\geq \kappa] \leq& \exp(-3\kappa/4). & \text{(\Cref{thm:X-pairs-uniform-tail-bound})}
    \end{align*}
    Therefore, we obtain \begin{align*}
        \Pr_{\pi \sim \algo D_X} [|X_\pi|= \kappa] \leq& \kappa \cdot \exp(-3\kappa/4) \\
        \leq& \exp\left(-\left(\frac{9-2\ln(6)}{12}\right)\kappa\right). & \text{(From $\kappa \geq 6$)}.
    \end{align*}
    Putting everything together and bounding the geometric series, we get that \begin{align*}
        \Pr_{\pi \sim \algo D_X} [|X_\pi|\geq \kappa] =& \sum_{i=\kappa}^{|X_1|} \Pr_{\pi \sim \algo D} [|X_\pi|= i] \\
        \leq& \sum_{i= \kappa}^{\infty} \exp\left(-\left(\frac{9-2\ln(6)}{12}\right)i\right) \\
        \leq& \left(1-\exp\Big(-\frac{9-2\ln(6)}{12}\Big)\right)^{-1} \exp\left(-\left(\frac{9-2\ln(6)}{12}\right) \kappa \right)\\
        \leq& 3\exp\left(\frac{-4 \kappa}{9}\right).
    \end{align*}
\end{proof}

\section{Query Problem Lower Bounds}
\label{sec:queryLowerBounds}

Unruh~\cite{Unruh2021,Unruh2023} proposed the task of ``double-sided zero-search''---a simple query problem dealing with invertible permutations which seems to go beyond the scope of current techniques, and may potentially offer new insights into the post-quantum security of the sponge construction. Unruh conjectured that it requires at least $\Omega(2^{n/2})$ many queries to solve the problem with constant success probability---and this is tight due to Grover's algorithm.

In this section, we prove query-lower bounds for Unruh's original double-sided zero-search problem, as well as a non-uniform variant that will be useful in proving the one-wayness of the single-round sponge with invertible permutations. We remark that our lower bounds on double-sided zero-search are tight up to a constant factor, and therefore resolve Unruh's original conjecture.



\subsection{Double-Sided Zero-Search}

In light of the difficulty in proving the post-quantum security of the sponge with invertible permutations,
Unruh~\cite{Unruh2021,Unruh2023} proposed the following simple query problem which seems beyond the scope of current techniques.

\begin{problem}[\textsc{Double-Sided Zero-Search}]
Given quantum query-access to a uniformly random permutation $\varphi : \{0,1\}^{2n} \rightarrow \{0,1\}^{2n}$ as well as its inverse $\varphi^{-1}$, output a pair of strings $x, y \in \{0,1\}^n$ such that $\varphi(x||0^n)=y||0^n$.   
\label{prob:doubleSidedZero}
\end{problem}
Recall that we call $x,y \in \bit^n$ which satisfy $\varphi(x||0^n)=y||0^n$ a \emph{zero pair} of $\varphi$. 

We prove that any algorithm must make at least $T = \Omega(\sqrt{\epsilon 2^n})$ many queries in order to find a zero pair for a random $2n$-bit permutation with probability $\epsilon$. At a high level, the first step of the proof is showing that there is some worst case instances on which the problem is hard for any fixed number of zero pairs, in \Cref{lem:doubleSidedZeroWorstCase}. Then we give a worst to average case reduction to show hardness for a random permutation with a fixed number of zero pairs in \Cref{lem:doubleSidedZeroAverageCase}. Finally, we give a bound for the general problem using the aforementioned bounds as well as tail bounds on the number of zero pairs.

\begin{theorem}
    Any quantum algorithm for \textsc{Double-Sided Zero-Search} that makes $T$ queries to a random $2n$-bit permutation or its inverse and succeeds with probability $\epsilon>0$ satisfies
    $$\epsilon \leq \frac{50(T+1)^2}{2^n}.$$  \label{thm:uniformZeroSearchHard}
\end{theorem}
\begin{proof}
    Let $\varphi$ be a uniformly random $2n$-bit permutation and $K=|Z_\varphi| \in \{0,1,\dots,2^n\}$ be the random variable corresponding to the number of zero pairs in $\varphi$. By the law of total probability, we can write the success probability $\epsilon$ of a $T$-query algorithm as
    \begin{align*}
        \epsilon =& \sum_{\kappa=0}^{2^n} \Pr[K=\kappa] \cdot \Pr[success | K=\kappa] \\
        =& \underbrace{\sum_{\kappa=1}^{6} \Pr[K=\kappa] \cdot \Pr[success | K=\kappa]}_{=P_1} + \underbrace{\sum_{\kappa=7}^{2^n} \Pr[K=\kappa] \cdot \Pr[success| K=\kappa]}_{=P_2}.
    \end{align*}
    We bound each term separately---note that splitting the sum at $\kappa=6$ is chosen to give the tightest bound. Beginning with the first, we get \begin{align*}
        P_1 \leq& \max_{\kappa \in [6]}\left\{\Pr[success | K=\kappa]\right\} & \text{(By Convexity)} \\
        \leq& \frac{48(T+1)^2}{2^n}. & \text{(By \Cref{lem:doubleSidedZeroAverageCase})}
    \end{align*}
    Continuing with the second, we find that \begin{align*}
        P_2 \leq& \sum_{\kappa=7}^{2^n} \Pr[success | K=\kappa] \cdot \exp(-3\kappa/4) & \text{(By \Cref{thm:X-pairs-uniform-tail-bound})} \\
        \leq& \sum_{\kappa=7}^{2^n} \frac{8(T+1)^2\kappa}{2^n} \cdot \exp(-3\kappa/4) & \text{(By \Cref{lem:doubleSidedZeroAverageCase})} \\
        \leq& \frac{2(T+1)^2}{2^n}. & \text{(Arithmetico-geometric series)}
    \end{align*}
    Therefore, we obtain the following bound on total success probability:  \begin{align*}
        \epsilon \leq \frac{50(T+1)^2}{2^n}
    \end{align*}
\end{proof}

We observe that this bound is tight up to constant factors for any $\epsilon$, as there is at least single zero pair with $\Omega(1)$ probability (\Cref{fact:Z-pairs-existence-probability}), which can be found using Grover's algorithm \cite{Grover96algorithm} with probability $\Theta(T^2/2^n)$ after $T$ queries. In particular, we have the following corollary.
\begin{corollary}\label{cor:grover-dszs}
    Given quantum oracle-access to a random $2n$-bit permutation and its inverse, the query complexity of outputting a zero pair with probability $\epsilon$ is $\Theta(\sqrt{\epsilon2^{n}})$.
\end{corollary}

\paragraph{Reduction from worst-case unstructured search.}
In the following lemma, we prove a query-lower bound for a worst-case variant of  \textsc{Double-Sided Zero-Search} in \Cref{prob:doubleSidedZero}, where we are given the promise that the underlying permutation has exactly $K$ zero pairs. We show the following.

\begin{lemma}
    Any $T$-query quantum algorithm for the \textsc{Double-Sided Zero-Search} problem succeeding with probability $\epsilon>0$ on a $2n$-bit worst-case permutation with exactly $K>0$ zero pairs must satisfy the inequality
    $$\epsilon \leq \frac{8(T+1)^2K}{2^n}.$$
    \label{lem:doubleSidedZeroWorstCase}
\end{lemma}
\begin{proof}
    We reduce the worst-case \textsc{Unstructured Search} problem with $K$ out of $2^n$ marked elements to a specific instance of \textsc{Double-Sided Zero-Search} for a $2n$-bit permutation $\varphi$ with exactly $K$ zero pairs without any query overhead.
    
    Suppose we are given quantum query access to a function $$f:\{0,1\}^n \rightarrow \{0,1\}$$ such that $|f^{-1}(1)|=K$. We now construct a permutation $\varphi:\{0,1\}^{2n} \rightarrow \{0,1\}^{2n}$, where $\varphi$ is defined as follows for any pair of strings $x, y \in \{0,1\}^n$:
    $$
    \varphi(x||y) =\begin{cases}
        x||y & \text{ if } f(x)=1 \\
        x||(y \oplus 1^n) & \text{ if } f(x)=0.
    \end{cases}
    $$
    It is apparent that we can implement a query to $\varphi$ using a single query to $f$. Further, we have  $\varphi=\varphi^{-1}$ since $\varphi$ is a composition of $1$ and $2$-cycles; hence, we can also implement queries to $\varphi^{-1}$. Now let $x||0^n$ be a zero pair of $\varphi$. It follows that $f(x)=1$, as otherwise $\varphi(x||0^n)=x||1^n$. Similarly, if, for any $x'$, we have $f(x')=1$, then $x'||0^n$ is a zero pair of $\varphi$. There are therefore $K$ zero pairs in $\varphi$ and, given a zero pair of $\varphi$, it is straight forward to find a preimage of $1$ under $f$, which can be used to solve \textsc{Unstructured Search}. The claim now follows from \Cref{cor:unstructuredSearchHardErrorVersion}.
\end{proof}

\paragraph{Worst-case to average-case reduction.}
In the following lemma, we prove a query-lower bound for an average-case variant of  \textsc{Double-Sided Zero-Search} in \Cref{prob:doubleSidedZero}, where we are given the promise that the underlying permutation has exactly $K$ zero pairs. In this case, average refers to a uniform random permutation over all that have exactly $K$ zero pairs. We show the following.

\begin{lemma}
    Any quantum algorithm for \textsc{Double-Sided Zero-Search} on a uniform random permutation $\varphi$, subject to the constraint that $\varphi$ has exactly $K>0$ zero pairs, making $T$ queries to $\varphi, \varphi^{-1}$ and succeeding with probability $\epsilon>0$ satisfies the inequality
    $$\epsilon \leq \frac{8(T+1)^2K}{2^n}$$
    \label{lem:doubleSidedZeroAverageCase}
\end{lemma}
\begin{proof}
Let $N=2^{2n}$ and let $Z \subset [N]$ denote the set of elements whose binary decomposition ends in $0^n$. Let $G_Z \leq S_N$ be the subgroup of permutations that fix $Z$, i.e.,
$$
G_Z \cong S_{Z} \times S_{[N] \setminus Z}.
$$
    We can re-randomize any worst-case permutation $\varphi : \bit^{2n} \rightarrow \bit^{2n}$ with $K$ zero pairs into an average-case permutation $\varphi^{\mathrm{sym}}$ with $K$ zero pairs as follows. \begin{enumerate}
        \item Randomly and independently sample permutations $\omega, \sigma \sim G_Z$.
        \item Define the symmetrized permutation $\varphi^{\mathrm{sym}} = \omega \circ \varphi \circ \sigma$.
    \end{enumerate}
    It follows from Lemma \ref{lemma:rerandomizationZeroSearch} that $\varphi^{\mathrm{sym}}$ is uniform random among permutations with $K$ zero pairs. Let $x, y \in Z$ be a zero pair of $\varphi^{\mathrm{sym}}$ such that $\varphi^{\mathrm{sym}}(x)=y$. We have that \begin{align}
        x' =& \, \sigma(x), & y'=\omega^{-1}(y)
    \end{align}
    satisfy $x', y' \in Z$ and $\varphi(x')=y'$, hence a zero pair of $\varphi$ can be constructed for free from a zero pair of $\varphi^{\mathrm{sym}}$. We can simulate queries to $\varphi^{\mathrm{sym}}$ as well as its inverse with a single query to $\varphi, \varphi^{-1}$ respectively, so the reduction incurs no overhead. Therefore, the claim follows from Lemma \ref{lem:doubleSidedZeroWorstCase}.
\end{proof}

\subsection{Alternative proof in the superposition-oracle framework}\label{sec:alternative}

In this section, we give an alternative proof for \Cref{lem:doubleSidedZeroAverageCase}, i.e., the query lower bound for \textsc{Double-Sided Zero-Search}. Because this lemma is a key component of \Cref{thm:uniformZeroSearchHard}, it presents an alternative approach for resolving \Cref{conj}.

We work with two-way accessible superposition oracles for invertible permutations.
The superposition oracle framework is a powerful tool which has been used to prove query lower bounds in a variety of settings, e.g.,~\cite{Ambainis_2011,Zhandry2018,rosmanis2022tight,Bai2022}.
Zhandry~\cite{Zhandry2018} also introduced the notion of compressed oracles as a means to ``record'' quantum queries to a random oracle. However, contrary to the work of Zhandry~\cite{Zhandry2018}, we do not need to ``compress'' the oracles and use inefficient representations instead. First, we introduce some relevant notation.

\paragraph{Superposition oracles.}
Let $\algo F_{n} = \{f: \bit^n \rightarrow \bit^n\}$ be the family of $n$-bit functions. The \emph{function register} for $f \in \algo F_n$ is a collection $F = \{F_x\}_{x \in \bit^n}$ with
$$
\ket{f}_F = \bigotimes_{x \in \bit^n} \ket{f(x)}_{F_x}.
$$
A query to $f$ in the superposition oracle framework amounts to the operation $O$ with
$$
O_{XYF} \ket{x}_X \ket{y}_Y \ket{f}_F = \mathsf{CNOT}_{F_x : Y} \ket{x}_X \ket{y}_Y \ket{f}_F = \ket{x}_X \ket{y \oplus f(x)}_Y \ket{f}_F.        
$$
We can also model inverse queries in the superposition oracle framework. This amounts to the operation $O^{-1}$ which is defined as the unitary
$$
O_{XYF}^{-1} \ket{x}_X \ket{y}_Y \ket{f}_F = \ket{x \oplus (\oplus_{x' : f(x') = y} \, x')}_X \ket{y}_Y \ket{f}_F.
$$
For example, if $f$ is a permutation, the inverse oracle amounts to the operation
$$
O_{XYF}^{-1} \ket{x}_X \ket{y}_Y \ket{f}_F = \ket{x \oplus f^{-1}(y)}_X \ket{y}_Y \ket{f}_F.
$$

\paragraph{Symmetrization.} Let $N=2^{2n}$ and consider a random permutation $\varphi \in S_N^\kappa$ which has exactly $\kappa$ zero pairs. In the superposition oracle framework, we can
model queries to such a permutation using a function register of the form
\begin{align}
\ket{\Phi_\kappa}_F = |S_N^\kappa|^{-\frac{1}{2}} \sum_{\varphi \in S_N^\kappa} \ket{\varphi}_F. \label{eq:F-register-kappa}
\end{align}
We now switch to an alternative characterization of the function registers which uses our technique of symmetrization via subset pairs of the symmetric group.
Let $Z \subset [N]$ denote the set of strings whose binary decomposition ends in $0^n$ many zeroes. Let $G_Z \leq S_N$ be the subgroup of permutations in $S_N$ that fix $Z$, i.e.,
$$
G_Z \cong S_{Z} \times S_{[N] \setminus Z}.
$$
Define the product group $G = G_Z \times G_Z$, where $G$ consists of pairs of permutations $(\sigma,\omega)$ such that each permutation fixes $Z$, i.e.,  $\sigma(Z)=Z=\omega(Z)$. Define the unitary representation $U: G \rightarrow \mathrm{GL}(\mathbb{C}^{\algo F})$ with $(\sigma,\omega) \mapsto U_{\sigma,\omega}$ and
$$
U_{\sigma,\omega} \ket{f}_F = \ket{\omega \circ f \circ \sigma}_F, \quad \text{ for } f \in \algo F_{2n}.
$$
Using symmetrization, we can now equivalently instantiate the function register $F$ in \Cref{eq:F-register-kappa} via the extension
$$
\ket{\Phi_\kappa}_F \mapsto \ket{\pi_\kappa^{\mathrm{sym}}}_{\Sigma \Omega F} = \frac{1}{\sqrt{|G|}} \sum_{(\sigma,\omega) \in G} \ket{\sigma}_\Sigma \ket{\omega}_\Omega U_{\sigma,\omega} \ket{\pi_\kappa}_F,
$$
where
$\pi_\kappa$ is an \emph{arbitrary} fixed permutation in $S_N^\kappa$ with exactly $\kappa$ zero pairs.
To see why this is a valid purification, we can use our symmetrization argument in~\Cref{lemma:rerandomizationZeroSearch} to argue that---once we trace out registers $\Sigma \Omega$---the reduced state equals
$$
 \frac{1}{|G|} \sum_{(\sigma,\omega) \in G} U_{\sigma,\omega} \proj{\pi_\kappa}_F U_{\sigma,\omega}^\dag = \frac{1}{|S_N^\kappa|} \sum_{\varphi \in S_N^\kappa} \proj{\varphi}_F.
$$

\paragraph{Double-sided zero-search revisited.}

To give an alternative proof for \Cref{lem:doubleSidedZeroAverageCase} in the superposition oracle framework, we first prove a lower bound for the decision problem.
We show that no quantum algorithm can distinguish whether it is querying a random invertible permutation with exactly $\kappa$ zero pairs, or a random invertible permutation with no zero pairs---unless it makes a large number of queries. Our proof is rooted in the hybrid argument~\cite{Bennett_1997} and one-way to hiding~\cite{cryptoeprint:2018/904}, and allows the query algorithm to query two oracles: one in the forward direction and one in the backward direction. We spell out the full details on the hybrid argument in terms of our notation for the convenience of the reader.

\begin{theorem}
Let $N=2^{2n}$ and let $\kappa \in \{0,1,\dots,\sqrt{N}\}$. Then, for any quantum algorithm $\algo D$ which makes at most $T$ many queries, it holds that
$$
\Big| 
\Pr_{\varphi \sim S_N^\kappa}[\algo D^{\varphi,\varphi^{-1}}(1^n) = 1] - \Pr_{\varphi \sim S_N^0}[\algo D^{\varphi,\varphi^{-1}}(1^n) = 1] 
\Big|   \,  \leq \, 2 T  \sqrt{\frac{\kappa}{2^n}}.
$$
\label{thm:decisionDSZSSuperposProof}
\end{theorem}
\begin{proof}
Suppose that a distinguisher $\algo D$ makes a total amount of $T$ oracle queries. In the superposition oracle framework, we can model the adversary/oracle interaction by introducing a function register $F$ which is outside of the view of $\algo D$.
Suppose the register $F$ is initially in the state $\ket{\Phi^0}_F$.
Using symmetrization, we can now equivalently instantiate the register $F$ as the uniform superposition
$$
\ket{\Phi^0}_F \mapsto \ket{\pi_0^{\mathrm{sym}} }_{\Sigma \Omega F} = \frac{1}{\sqrt{|G|}} \sum_{(\sigma,\omega) \in G} \ket{\sigma}_\Sigma \ket{\omega}_\Omega U_{\sigma,\omega} \ket{\pi_0}_F,
$$
where
$\pi_0(x_1||x_2) = (x_1||x_2 \oplus 1^n)$ is a fixed permutation in $S_N^0$ with no zero pairs.

For the remainder of the proof, we will analyze the effect of swapping out the symmetrization of
$\pi_0 \in S_N^0$
with symmetrization of a different permutation which comes from $S_N^\kappa$.
Specifically, we will choose a particular permutation $\pi_K$ which is generated as follows: first sample a random subset $K
\subset \bit^n$ (independently of everything else) of size $|K|=\kappa$, and then let $\pi_K$ be the permutation
$$
\pi_K(x_1||x_2) = 
\begin{cases}
x_1||x_2 & \text{ if } x_1 \in K;\\
x_1||x_2 \oplus 1^n & \text{ otherwise. }
\end{cases}
$$
Note that $\pi_K$ differs from $\pi_0$ on exactly $\kappa \cdot 2^n$ many inputs of the form $(x_1||x_2)$, where $x_1 \in K$ and $x_2$ is arbitrary. Moreover, $\pi_K$ has exactly $\kappa$ zero pairs of the form $(x_1||0^n)$ whenever $x_1 \in K$. Thus, $\pi_K \in S_N^\kappa$. Notice that we can map between the symmetrizations $\ket{\pi_0^{\mathrm{sym}} }$ and $\ket{\pi_K^{\mathrm{sym}}}$ via the following unitary
$$
\mathsf{SWAP}^{\pi_K,\pi_0}_{\Sigma \Omega F}
=U_{\Sigma \Omega:F} (\Id_{\Sigma \Omega} \otimes W^K_F) U_{\Sigma \Omega:F}^\dag\, ,
$$
where $U_{\Sigma \Omega:F}$ is a controlled unitary of the form 
$$
U \ket{\sigma}_\Sigma \ket{\omega}_\Omega \ket{f}_F = \ket{\sigma}_\Sigma \ket{\omega}_\Omega  U_{\sigma,\omega}\ket{f}_F, \quad \text{ for } f \in \algo F_{2n},
$$
and where $W^K$ is the unitary which flips between $\ket{\pi_K}$ and $\ket{\pi_0}$. \footnote{$W^K$ has a well-defined extension to all $\ket{f}_F$ since we just need to $XOR$ the string $(0^n||1^n)$ into the function output whenever the input is of the form $(x_1||y)$, for $x_1 \in K$ and $y \in \bit^n$}
Therefore,
$$
\mathsf{SWAP}^{\pi_K,\pi_0}_{\Sigma \Omega F}
\ket{\pi_0^{\mathrm{sym}} }_{\Sigma \Omega F}= \ket{\pi_K^{\mathrm{sym}}}_{\Sigma \Omega F}.
$$ 
Our goal is to show the following property:
on average over the choice of $K$, the query magnitude on symmetrized oracle calls that involve $\pi_0$ on $(x_1||x_2)$ with $x_1 \in K$ and arbitrary $x_2$ will be small. 
We proceed via a standard hybrid argument which involves a sequence of intermediate states.
Note that oracle calls to $\ket{\omega \circ \pi_0 \circ \sigma}_F$, where the permutations $\sigma,\omega$ come from registers $\Sigma$ and $\Omega$, only involve $K$ in one of two ways:
\begin{itemize}
    \item a forward query in the $X$ register is made on $\sigma^{-1}(x_1||x_2)$ with $x_1 \in K$; or

    \item a backward query in the $Y$ register is made on $\omega(x_1||x_2)$ with $x_1 \in K$.
\end{itemize}
Define the following projectors that exclude each of the two cases:
\begin{align}
\Pi_{X \Sigma}^K &= \sum_{x' \in \bit^{2n}} \proj{x'}_X \otimes \bigotimes_{\substack{x_1 \in K\\
x_2 \in \bit^n}} \big(\Id - \proj{x_1||x_2} \big)_{\Sigma_{x'}}\\
\Xi_{Y\Omega}^{K} &= \sum_{y' \in \bit^{2n}} \proj{y'}_Y \otimes 
\bigotimes_{\substack{x_1 \in K\\
x_2 \in \bit^n}} 
\big(\Id - \proj{y'} \big)_{\Omega_{x_1||x_2}} \, .
\end{align}
If we apply the projectors $\Pi_{X \Sigma}^K$ and $\Xi_{Y \Omega}^K$ to each forward/backward query just before the oracle is applied, then swapping from $\pi_0$ to $\pi_K$ before the oracle evaluation has no effect. Formally, we have the commutation relations\footnote{For two linear operators $A$ and $B$, the commutator is defined as $[A,B] := AB-BA$.}
\begin{align}
\left[\mathsf{SWAP}^{\pi_K,\pi_0}_{\Sigma \Omega F}, O_{XY F} \,\Pi_{X \Sigma}^K \right] &=0 \label{eq:comm-relations}\\
\left[\mathsf{SWAP}^{\pi_K,\pi_0}_{\Sigma \Omega F}, O_{XY F}^{-1} \,\Xi_{Y\Omega}^{K}\right] &=0.\label{eq:comm-relations-II}
\end{align}
For $i \in [T]$, we now define the ensemble of projectors $\{P_i^K\}_{i \in [T]}$, where
\begin{align}\label{eq:def-Pi}
P_i^K = 
\begin{cases}
\Pi^K, & \text{ if } O_i = O;\\
\Xi^K, & \text{ if } O_i = O^{-1}.
\end{cases}
\end{align}
Suppose $F$ is initially in the state $\ket{\Phi^0}_F$ and $\algo D$ starts out with some normalized pure state $\ket{\eta_0}$ on registers $XYE$.
We can model the interaction between $\algo D$ and the oracle as sequence of oracle queries and unitaries such that for $i \in [T]$,
\begin{align}
\ket{\psi_i}_{XYE\Sigma \Omega F} = U_i O_i U_{i-1} O_{i-1} \cdots U_1 O_1 \ket{\psi_0} \, .
\end{align}
Here, the unitaries $O_i \in \{O,O^{-1}\}$ represent a call to the forward/backward oracles acting on registers $XYF$, the $U_i$ represent intermediate unitaries performed by $\algo D$ which act on $XYE$, for some workspace register $E$, and where $\ket{\psi_0}$ is the state 
$$
\ket{\psi_0}_{XYE\Sigma \Omega F} =  \ket{\eta_0}_{XYE} \otimes \ket{\pi_0^{\mathrm{sym}} }_{\Sigma \Omega F} \, .
$$
To study the effects of swapping between symmetrizations of
$\pi_0 \in S_N^0$
and $\pi_K \in S_N^\kappa$, we now define $i$-th intermediate state as
$$
\ket{\psi_i^K}_{XYE\Sigma \Omega F} = U_i O_i U_{i-1} O_{i-1} \cdots U_1 O_1 \ket{\psi_0^K} \, ,
$$
where $\ket{\psi_0^K}$ is the initial state corresponding to a symmetrization of $\pi_K$, i.e.,
$$
\ket{\psi_0^K}_{XYE\Sigma \Omega F} =  \ket{\eta_0}_{XYE} \otimes \ket{\pi_K^{\mathrm{sym}} }_{\Sigma \Omega F} \, .
$$
We now introduce the $K$-dependent oracle unitaries $\{O^K_i\}_{i \in [T]}$ where
$$
O^K_i = \left(\mathsf{SWAP}^{\pi_K,\pi_0}_{\Sigma \Omega F} \right)^\dag (O_i)_{XYF} \,\, \mathsf{SWAP}^{\pi_K,\pi_0}_{\Sigma \Omega F}
$$
acts on registers $XY\Sigma \Omega F$. This yields the following equivalent characterization
\begin{align}
\ket{\psi_i^K}_{XYE\Sigma \Omega F} = U_i O_i^K U_{i-1} O_{i-1}^K \cdots U_1 O_1^K \ket{\psi_0} \, .
\end{align}
By unitarity, we can now bound the distance between the $i$-th hybrid states as
\begin{align}
\left\| \ket{\psi_i} -\ket{\psi_i^K}   \right\|^2 &=  \left\| U_i O_i \ket{\psi_{i-1}} -U_i O_i^K \ket{\psi_{i-1}^K}   \right\|^2 \nonumber\\
&=  \left\| O_i \ket{\psi_{i-1}} - O_i^K \ket{\psi_{i-1}^K}   \right\|^2 \nonumber\\
&=  \left\| O_i \ket{\psi_{i-1}} -  O_i^K \ket{\psi_{i-1}} + O_i^K \ket{\psi_{i-1}} - O_i^K \ket{\psi_{i-1}^K}   \right\|^2 \nonumber\\
&\leq \left\| O_i \ket{\psi_{i-1}} -  O_i^K \ket{\psi_{i-1}} \right\|^2 + \left\|O_i^K \ket{\psi_{i-1}} - O_i^K \ket{\psi_{i-1}^K}   \right\|^2 \nonumber\\
& \quad\quad +2\cdot \left\| O_i \ket{\psi_{i-1}} -  O_i^K \ket{\psi_{i-1}} \right\| \cdot \left\| O_i^K (\ket{\psi_{i-1}} - \ket{\psi_{i-1}^K})\right\| \nonumber\\
&= \left\| (O_i -O_i^K)\ket{\psi_{i-1}} \right\|^2 + \left\|\ket{\psi_{i-1}} - \ket{\psi_{i-1}^K}  \right\|^2 \nonumber\\
& \quad\quad +2\cdot \left\| (O_i -O_i^K)\ket{\psi_{i-1}} \right\| \cdot \left\| \ket{\psi_{i-1}} - \ket{\psi_{i-1}^K}\right\|, \label{eq:after-CS}
\end{align}
where the triangle inequality gives the inequality. We now observe that
\begin{align}
\left\|(O_i -O_i^K)\ket{\psi_{i-1}} \right\| &= \left\|(O_i -O_i^K) \bar{P}_i^K \ket{\psi_{i-1}} + (O_i -O_i^K) P_i^K \ket{\psi_{i-1}} \right\| \nonumber \\
&\leq  \left\|(O_i -O_i^K) \bar{P}_i^K \ket{\psi_{i-1}}\right\| + \left\|(O_i -O_i^K) P_i^K \ket{\psi_{i-1}} \right\| \nonumber \\
&\leq 2 \cdot \left\|\bar{P}_i^K \ket{\psi_{i-1}}\right\| + \left\|(O_i -O_i^K) P_i^K \ket{\psi_{i-1}} \right\| \nonumber\\
&= 2 \cdot \left\|\bar{P}_i^K \ket{\psi_{i-1}}\right\|. \label{eq:CR}
\end{align}
Here, the second to last line uses that $(O_i -O_i^K)$ has operator norm at most $2$. For the last line, we first invoke the fact that $\mathsf{SWAP}^{\pi_K,\pi_0}_{\Sigma \Omega F}$ and $P_i^K$ commute, and then use the commutation relations from \Cref{eq:comm-relations,eq:comm-relations-II} so that
\begin{align*}
\left\|(O_i -O_i^K) P_i^K \ket{\psi_{i-1}} \right\| &= 
\left\|\mathsf{SWAP}^{\pi_K,\pi_0}_{\Sigma \Omega F}(O_i -O_i^K) P_i^K \ket{\psi_{i-1}} \right\| \quad\quad \text{(By unitarity)}\\
&=\left\|(\mathsf{SWAP}^{\pi_K,\pi_0}_{\Sigma \Omega F} \, O_i  \, P_i^K - O_i \,\mathsf{SWAP}^{\pi_K,\pi_0}_{\Sigma \Omega F} \, P_i^K)\ket{\psi_{i-1}} \right\|\\
&=\left\|(\mathsf{SWAP}^{\pi_K,\pi_0}_{\Sigma \Omega F} \, O_i  \, P_i^K - O_i \, P_i^K \, \mathsf{SWAP}^{\pi_K,\pi_0}_{\Sigma \Omega F})\ket{\psi_{i-1}} \right\|\\
&=\left\|\left[\mathsf{SWAP}^{\pi_K,\pi_0}_{\Sigma \Omega F}, \, O_i \,P_i^K \right]\ket{\psi_{i-1}} \right\| =0.
\end{align*}
By combining \Cref{eq:after-CS} and \Cref{eq:CR}, we get that
\begin{align}
\left\| \ket{\psi_i} -\ket{\psi_i^K}   \right\|^2 &\leq 4 \cdot \left\|\bar{P}_i^K \ket{\psi_{i-1}}\right\|^2 + \left\|\ket{\psi_{i-1}} - \ket{\psi_{i-1}^K}  \right\|^2 \nonumber\\
& \quad\quad +4\cdot \left\|\bar{P}_i^K \ket{\psi_{i-1}}\right\| \cdot \left\| \ket{\psi_{i-1}} - \ket{\psi_{i-1}^K}\right\|    \nonumber \\
&= \left( \left\| \ket{\psi_{i-1}} - \ket{\psi_{i-1}^K}\right\| + 2 \cdot \left\|\bar{P}_i^K \ket{\psi_{i-1}}\right\| \right)^2,
\end{align}
and thus
\begin{align}
\left\| \ket{\psi_T} -\ket{\psi_T^K}   \right\| \leq 2 \cdot \sum_{i=1}^T \left\|\bar{P}_i^K \ket{\psi_{i}}\right\|.
    \label{eq:hybrid-distance}
\end{align}
Therefore, we can bound the average distance as follows:
\begin{align*}
&\mathsf{TD} \Big(\proj{\psi_T}, \underset{\substack{K\subset \bit^n\\
|K|=\kappa
}}{\E} \proj{\psi_T^K}\Big)\\
&\quad \leq
\underset{\substack{K\subset \bit^n\\
|K|=\kappa
}}{\E} \mathsf{TD} \Big(\proj{\psi_T}, \proj{\psi_T^K}\Big)  && \text{(By convexity)}\\
&\quad \leq
\underset{\substack{K\subset \bit^n\\
|K|=\kappa
}}{\E} \left\|\ket{\psi_T} - \ket{\psi_T^K}\right\| && \text{(Norm inequality)}\\
&\quad \leq
\underset{\substack{K\subset \bit^n\\
|K|=\kappa
}}{\E} 2 \cdot \sum_{i=1}^T \left\|\bar{P}_i^K \ket{\psi_{i}}\right\| && \text{(By \Cref{eq:hybrid-distance})}\\
&\quad = 2 T\underset{i \sim [T]}{\E} \,\,\underset{\substack{K\subset \bit^n\\
|K|=\kappa
}}{\E} \left\|\bar{P}_i^K \ket{\psi_{i}}\right\|\\
&\quad \leq 2 T \sqrt{\underset{i \sim [T]}{\E} \,\,\underset{\substack{K\subset \bit^n\\
|K|=\kappa
}}{\E} \left\|\bar{P}_i^K \ket{\psi_{i}}\right\|^2}.&& \text{(Jensen's inequality)}.
\end{align*}

To complete the proof, it suffices to bound the expectation.
By the definition of $\{P_i^K\}_{i \in [T]}$ in \Cref{eq:def-Pi}, we have to consider two cases for each $i\in [T]$. First, suppose that the $i$-th query is a forward query with $O_i = O$. Then,
\begin{align}
\label{eq:bad-event-I}
\underset{\substack{K\subset \bit^n\\
|K|=\kappa
}}{\E} \left\|\bar{P}_i^K \ket{\psi_{i}}\right\|^2 = \underset{\substack{K\subset \bit^n\\
|K|=\kappa
}}{\E} \left[\left\| (\Id-\Pi^K)_{X \Sigma} \ket{\psi_i}_{XY\Sigma\Omega E}\right\|_2^2 \right]  \leq \frac{\kappa}{2^n}.
\end{align}
Next, suppose that the $i$-th query is to the inverse oracle $O_i = O^{-1}$. Then, we have
\begin{align}
\label{eq:bad-event-II}
\underset{\substack{K\subset \bit^n\\
|K|=\kappa
}}{\E} \left\|\bar{P}_i^K \ket{\psi_{i}}\right\|^2 = \underset{\substack{K\subset \bit^n\\
|K|=\kappa
}}{\E} \left[\left\| (\Id-\Xi^K)_{Y \Omega} \ket{\psi_i}_{XY\Sigma\Omega E}\right\|_2^2 \right]  \leq \frac{\kappa}{2^n}.
\end{align}
Putting everything together and using \Cref{eq:bad-event-I} and \Cref{eq:bad-event-II}, we arrive that
$$
\mathsf{TD} \Big(\proj{\psi_T}, \underset{\substack{K\subset \bit^n\\
|K|=\kappa
}}{\E} \proj{\psi_T^K}\Big) \leq 2 T\cdot \sqrt{\frac{\kappa}{2^n}}.
$$
This proves the claim.
\end{proof}

\paragraph{Alternative to resolving Unruh's conjecture.}
We can use \Cref{thm:decisionDSZSSuperposProof} to prove an alternative version of \Cref{lem:doubleSidedZeroAverageCase}, which in turn is the key component of \Cref{thm:uniformZeroSearchHard}, and thus in resolving \Cref{conj}.

\begin{lemma}[Alternative to \Cref{lem:doubleSidedZeroAverageCase}]
    \textit{Any quantum $T$-query algorithm for the \textsc{Double-Sided Zero-Search} problem with respect to a random $2n$-bit invertible permutation with exactly $\kappa>0$ zero pairs, $\pi \sim S_{2^{2n}}^\kappa$, that succeeds with probability $\epsilon>0$ satisfies the inequality
    $$\epsilon \leq 2(T+1) \sqrt{\frac{\kappa}{2^{n}}}.$$}
    \label{lem:DSZSAverageCaseSuperpos}
\end{lemma}

Before proceeding with the proof, we observe a difference between \Cref{lem:DSZSAverageCaseSuperpos} and \Cref{lem:doubleSidedZeroAverageCase}. While both are tight up to constant factors for a constant success probability (and thus sufficient to resolve \Cref{conj}), the bound in \Cref{lem:DSZSAverageCaseSuperpos} is of the form $\epsilon = O(T \sqrt{\kappa/2^n})$ whereas in \Cref{lem:doubleSidedZeroAverageCase} the bound derived is of the form $\epsilon = O(T^2 \kappa/2^n)$. The latter bound is quadratically tighter for small $\epsilon \ll 1$, and so we state our main theorem in these terms. We leave tightening the analysis in the superposition-oracle framework as an interesting open problem.

\begin{proof}
    Suppose that $\algo A^{\varphi, \varphi^{-1}}$ finds a zero pair with probability $\epsilon > 0$ whenever $\varphi$ is random $2n$-bit permutation with exactly $\kappa$ zero pairs. We can transform such an $\algo A$ into a distinguisher between permutations with $\kappa$ zero pairs and no zero pairs by running it on a given permutation, and guessing that $\varphi$ has $\kappa$ zero pairs if and only if $\algo A$ finds a zero pair. Formally, given a permutation $\varphi,\varphi^{-1}$ the distinguisher $\algo D^{\varphi, \varphi^{-1}}$ proceeds as follows:\begin{enumerate}
        \item Run $\algo A^{\varphi, \varphi^{-1}}(1^n)$ to obtain\footnote{If the output is not of this form, simply output $0$ (guess $\varphi$ has no zero pairs).} a pair $(x||0^n, y||0^n)$.
        \item If $\varphi(x||0^n)=y||0^n$, then output $1$ (guess that $\varphi$ has $\kappa$ zero pairs)
        \item Otherwise, output $0$ (guess that $\varphi$ has no zero pairs)
    \end{enumerate}
    If $\algo A$ makes $T$ queries then $\algo D$ makes $T+1$ queries, and the distinguishing advantage of $\algo D$ is the probability that $\algo A$ succeeds given that $\varphi$ has $\kappa$ zero pairs. We have \begin{align*}
        \Big| 
        \Pr_{\varphi \sim S_N^\kappa}[\algo D^{\varphi,\varphi^{-1}}(1^n) = 1] - \Pr_{\varphi \sim S_N^0}[\algo D^{\varphi,\varphi^{-1}}(1^n) = 1] 
        \Big|   \,  = \, \Pr_{\varphi \sim S_N^\kappa}[\algo A^{\varphi, \varphi^{-1}}(1^n)].
    \end{align*}
    Now applying \Cref{thm:decisionDSZSSuperposProof}, we obtain \begin{align*}
        \Pr_{\varphi \sim S_N^\kappa}[\algo A^{\varphi, \varphi^{-1}}(1^n)] \, \leq \, 2(T+1) \sqrt{\frac{\kappa}{2^n}}.
    \end{align*}
\end{proof}

\subsection{Non-Uniform Double-Sided Search}

In this section we define a new variant of
Unruh's original double-sided zero-search problem~\cite{Unruh2021,Unruh2023}. Namely, we consider the
\textsc{Double-Sided Search} problem in which the distribution on permutations is \emph{non-uniform}, and the constraints on preimages versus images are more flexible.
This problem will be useful for proving the one-wayness of the single-round sponge, which is our motivation for studying it.

\paragraph{Notation.} For the remainder of this section, we fix the following notation:
\begin{itemize}
    \item $N=2^{n}$ and $r, c>0$ are integers such that $r+c=n$.
    \item $X=(X_1,X_2)$ is a pair of subsets, where $X_1 \subset \bit^n$ is the set of bitstrings ending in $0^c$, $X_2 \subset \bit^n$ is the set of bitstrings beginning in $0^r$.
    \item $|X_\varphi|$ denotes the number of $X$-pairs with respect to a permutation $\varphi \in S_N$.
    \item $G_1$ is the subgroup of $S_N$ consisting of permutations which preserve membership in $X_1$, whereas $G_2$ is the subgroup of $S_N$ which preserves membership in $X_2$ 
\end{itemize}

Note that $|X_1|\cdot |X_2|=N$, and $G_1, G_2$ are both Young subgroups. In this section, we consider the distribution $\algo D_X$ from \Cref{defn:nonUniformDist} which assigns more weight to permutations $\varphi \in S_N$ that have a large amount of $X$-pairs.
We will occasionally omit the subscript which specifies $X=(X_1,X_2)$, for the pair of subsets $X_1,X_2$, when it is clear from context. 

\begin{problem}[\textsc{Non-uniform Double-Sided Search}]
Given quantum query access to a permutation $\varphi : \{0,1\}^{n} \rightarrow \{0,1\}^{n}$ as well as its inverse $\varphi^{-1}$, where $\varphi$ is sampled according to the non-uniform distribution $\algo D_X$ from \Cref{defn:nonUniformDist}, output a pair of strings $x \in X_1, y \in X_2$ such that $\varphi(x)=y$. 
\label{prob:nonUniformZeroSearch}
\end{problem}

In the remainder of this section, we prove that any algorithm must make at least $T = \Omega(\sqrt{\epsilon 2^{\min(r,c)}})$ many queries in order to find an $X$-pair for an $n$-bit permutation sampled from $\algo D_X$ with probability $\epsilon$. At a high level, the first step of the proof is showing that there exist worst-case instances on which the problem is hard for any fixed number of $X$-pairs (see \Cref{lem:doubleSidedNonuniformWorstCase}). Then, we give a worst-case to average-case reduction to show hardness for a random permutation with a fixed number of $X$-pairs in \Cref{lem:doubleSidedNonuniformAverageCase}. In the below theorem we give a bound for the general problem using the aforementioned results, as well as tail bounds on the number of subset pairs derived in \Cref{sec:combinatoricsSubsetPairs}.

\begin{theorem}
    Any quantum algorithm for \textsc{Non-uniform Double-Sided Search} that makes $T$ queries to an invertible permutation and succeeds with probability $\epsilon>0$ satisfies
    $$\epsilon \leq \frac{80(T+1)^2}{2^{\min(r,c)}}.$$
    \label{thm:nonUniformZeroSearchHard}
\end{theorem}
\begin{proof}
    Let $\varphi\sim\algo D_X$ be a permutation sampled according to the distribution $\algo D_X$ from \Cref{defn:nonUniformDist} and let $K=|X_\varphi| \in \{0,1,\dots,2^{\min(r,c)}\}$ be the random variable corresponding to the number of $X$-pairs in $\varphi$. By the law of total probability, we can write the success probability of any $T$-query algorithm as the sum
    \begin{align*}
        \epsilon =& \sum_{\kappa=0}^{2^{\min(r,c)}} \Pr[K=\kappa] \cdot \Pr[success| K=\kappa] \\
        =& \underbrace{\sum_{\kappa=1}^{6} \Pr[K=\kappa] \cdot \Pr[success| K=\kappa]}_{=P_1} + \underbrace{\sum_{\kappa=7}^{2^{\min(r,c)}} \Pr[K=\kappa] \cdot \Pr[success| K=\kappa]}_{=P_2}.
    \end{align*}
    We bound each term separately---note that splitting the sum at $\kappa=6$ is chosen to give the tightest bound. Beginning with the first, we find \begin{align*}
        P_1 \leq& \max_{\kappa \in [6]}\left\{\Pr[success| K=\kappa]\right\} & \text{(By Convexity)} \\
        \leq& \frac{48(T+1)^2}{2^{\min(r,c)}}. & \text{(By \Cref{lem:doubleSidedNonuniformAverageCase})}
    \end{align*}
    Next, continuing with the second term, we get \begin{align*}
        P_2 \leq& \sum_{\kappa=7}^{2^{\min(r,c)}} \Pr[success| K=\kappa] \cdot 3\exp(-4\kappa/9) & \text{(By \Cref{thm:X-pairs-nonuniform-tail-bound})} \\
        \leq& \sum_{\kappa=7}^{2^{\min(r,c)}} \frac{24(T+1)^2 \kappa}{2^{\min(r,c)}}\cdot\exp(-4\kappa/9) & \text{(By \Cref{lem:doubleSidedNonuniformAverageCase})} \\
        \leq& \frac{32(T+1)^2}{2^{\min(r,c)}}. & \text{(Arithmetico-geometric series)}
    \end{align*}
    This concludes the proof, as we can bound the success probability as
    \begin{align*}
        \epsilon \leq \frac{80(T+1)^2}{2^{\min(r,c)}}.
    \end{align*}
\end{proof}

This bound is tight up to constant factors for any $\epsilon$, as there is at least a single $X$-pair with $\Omega(1)$ probability (\Cref{fact:Z-pairs-existence-probability}), which can be found using Grover's algorithm \cite{Grover96algorithm} with probability $\Theta(T^2/2^{\min(r, c)})$ after $T$ queries. In particular, we have the following corollary.
\begin{corollary}\label{cor:grover-nonuniform-dszs}
    Given quantum oracle-access to an $n$-bit permutation chosen from $\algo D_X$, the query complexity of outputting an $X$-pair with probability $\epsilon$ is $\Theta(\sqrt{\epsilon 2^{\min(r, c)}})$.
\end{corollary}

\paragraph{Reduction from worst-case unstructured search.}
In the following lemma, we prove a query-lower bound for a worst-case variant of \textsc{Non-uniform Double-Sided Search} in \Cref{prob:nonUniformZeroSearch}, where we are given the promise that the underlying permutation has exactly $K$ many $X$-pairs. We show the following.

\begin{lemma}
    Any $T$-query quantum algorithm for \textsc{Non-Uniform Double-Sided Search} on a worst-case permutation with exactly $K>0$ many $X$-pairs that succeeds with probability $\epsilon>0$ must satisfy the inequality
    $$\epsilon \leq \frac{8(T+1)^2K}{2^{\min(r,c)}}.$$
    \label{lem:doubleSidedNonuniformWorstCase}
\end{lemma}
\begin{proof}
    We can reduce \textsc{Unstructured Search} with $K$ out of $2^{\min(r, c)}$ marked elements to a specific instance of \textsc{Non-Uniform Double-Sided Search} with respect to a permutation $\varphi$ on $N=2^{n}$ elements and $K$ many $X$ pairs. Moreover, as we will show, our reduction does not incur any query overhead. 
    
    Suppose we are given query access to a function $$f:\{0,1\}^{{\min(r, c)}} \rightarrow \{0,1\}$$ such that $|f^{-1}(1)|=K$. We can construct a permutation $\varphi:\{0,1\}^{n} \rightarrow \{0,1\}^{n}$ for $x\in \bit^{\min(r, c)}, y \in \bit^{\max(r, c)}$ as follows, where $s^R$ denotes the reverse of string $s$:
    $$
    \varphi(x||y) =\begin{cases}
        (x||y)^R & \text{ if } f(x)=1 \\
        (x||y)^R \oplus 1^r||0^c  & \text{ if } f(x)=0.
    \end{cases}
    $$

    It is apparent that we can implement a query to $\varphi$ using a single query to $f$. Moreover, the image $\varphi(x||y)$ always contains $x^R$ in the last $\min(r, c)$ bits, which suffices to determine $f(x)$, and thus to determine $x||y$ by inverting whichever case was applied. Hence, we can implement $\varphi^{-1}$ with a single query to $f$. Now let $x \in \bit^{\min(r, c)}, y\in\bit^{\max(r,c)}$ such that $y$ ends in $0^c$ and $\varphi(x||y)$ begins with $0^r$. It follows that $f(x)=1$ and $y=0^{\max(r, c)}$, as otherwise either $y$ does not end in $0^c$ or $\varphi(x||y)$ does not begin with $0^r$. Similarly, if for any $x'$ we have $f(x')=1$ then $x'||0^{\max(r, c)}$ is an $X$-pair of $\varphi$. Therefore, there are $K$ many $X$-pairs in $\varphi$, and given an $X$-pair of $\varphi$ it is straight forward to find a preimage of $1$ under $f$, which can be used to solve \textsc{Unstructured Search}. The claim now follows from \Cref{cor:unstructuredSearchHardErrorVersion}.
\end{proof}

\paragraph{Worst-case to average-case reduction.}
In the following lemma, we prove a query-lower bound for an average-case variant of  \textsc{Non-uniform Double-Sided Search} in \Cref{prob:nonUniformZeroSearch}, where we are given the promise that the underlying permutation has exactly $K$ many $X$-pairs. In this case, average refers to a uniform random permutation subject to this promise. We show the following.

\begin{lemma}
    Any $T$-query algorithm for \textsc{Non-Uniform Double-Sided Search} on a random permutation which has exactly $K>0$ many $X$-pairs that succeeds with probability $\epsilon>0$ must satisfy the inequality
    $$\epsilon \leq \frac{8(T+1)^2K}{2^{\min(r,c)}}.$$
    \label{lem:doubleSidedNonuniformAverageCase}
\end{lemma}
\begin{proof}
    We can re-randomize any worst-case permutation $\varphi$ with $K$ many $X$-pairs toward an average-case permutation $\varphi^{\mathrm{sym}}$ with $K$ many $X$-pairs as follows. \begin{enumerate}
        \item Randomly and independently sample $\omega \sim G_1, \sigma \sim G_2$, where $\omega$ preserves membership in $X_1$ and $\sigma$ preserves membership in $X_2$.
        \item Define the symmetrized permutation $\varphi^{\mathrm{sym}} = \omega \circ \varphi \circ \sigma$.
    \end{enumerate}
    It follows from Lemma \ref{lemma:rerandomizationZeroSearch} that $\varphi^{\mathrm{sym}}$ is random among permutations with $K$ many $X$-pairs. Let $(x, y) \in X_1 \times X_2$ be an $X$-pair of $\varphi^{\mathrm{sym}}$ such that $\varphi^{\mathrm{sym}}(x)=y$. Then, 
    \begin{align}
        x' =& \, \sigma(x), & y'=\omega^{-1}(y)
    \end{align}
    which satisfy $x' \in X_1, y' \in X_2$ and $\varphi(x')=y'$.  Hence, an $X$-pair of $\varphi$ can be constructed for free from an $X$-pair of $\varphi^{\mathrm{sym}}$. We can simulate queries to $\varphi^{\mathrm{sym}}$ as well as its inverse with a single query to $\varphi, \varphi^{-1}$ respectively, so the reduction incurs no overhead. Lemma \ref{cor:unstructuredSearchHardErrorVersion} now implies the claim.
\end{proof}




\section{One-Wayness of the Sponge Construction}
\label{sec:spongeOneWay}

In this section, we prove our main result; namely, the quantum one-wayness of the single-round sponge when instantiated with invertible random permutations.

\paragraph{The sponge construction.}
Recall that the sponge construction~\cite{KeccakSponge3} uses two main parameters which we call the \emph{rate} $r$ and the \emph{capacity} $c$ of the sponge. The internal state of the sponge function gets updated through successive applications of a \emph{block function} $\varphi: \bit^{r+c} \rightarrow \bit^{r+c}$.
Let $\sigma^{(0)} = (0^r || \mathsf{IV})$ denote the initial state of the sponge, where $\mathsf{IV} \in \bit^c$ is some initialization vector; for simplicity, we consider the choice $\mathsf{IV}=0^c$. In the following, we use the notation $x_r \in \bit^r$ and $x_c \in \bit^c$ to denote the first $r$ bits as well as the last $c$ bits of $x \in \bit^{r+c}$, respectively.

The sponge hash function $\mathsf{Sp}^\varphi :  \bit^* \rightarrow \bit^r$ is evaluated as follows on  a sequence of $r$-bit blocks $m = (m_1,\dots,m_\ell)$:
\begin{itemize}
    \item (Absorption phase) for $i \in [\ell]$, update the state by computing
    $$\sigma^{(i)}  = \varphi\big((\sigma^{(i-1)}_r \oplus m_i)|| \sigma^{(i-1)}_c \big).$$
    \item (Squeezing phase) parse the final state $\sigma^{(\ell)}$ as $(\sigma_r^{(\ell)},\sigma_c^{(\ell)}) \in \bit^r \times \bit^c$, and output the digest $\mathsf{Sp}^\varphi(m_1,\dots,m_\ell) = \sigma_r^{(\ell)}$ consisting of the first $r$ bits of $\sigma^{(\ell)}$.
\end{itemize}
We remark that the sponge construction also allows for variable-length output by essentially running the absorption phase in reverse~\cite{KeccakSponge3}; we choose to ignore this option for simplicity as we only focus on single-round sponge hashing.

\paragraph{One-wayness of the single-round sponge.} 
Recall that, in the special case when there is only a single round of absorption, the sponge function $\mathsf{Sp}^\varphi: \bit^{r} \rightarrow \bit^r$ takes on a simple form; namely, on input $x \in \bit^r$, the output is given by the string $y = \mathsf{Sp}^\varphi(x)$, where $y$ corresponds to the first $r$ bits of $\varphi(x||0^c)$. 

We now answer the following question on the one-wayness of the single-round sponge; namely, how many queries does a quantum algorithm $\algo A^{\varphi,\varphi^{-1}}$ need to invert $\mathsf{Sp}^\varphi$ on a randomly chosen input, when $\varphi$ is an invertible random permutation? 

Let us first define the one-wayness experiment in full generality. Note that the sponge function $\mathsf{Sp}^\varphi$ depends on $\varphi$, which in turn is modeled as a uniformly random permutation which we have oracle access to. Therefore, we need to consider a variant of one-wayness which is parameterized by an ensemble of oracle functions. 

\begin{definition}[One-wayness] 
Let $\lambda \in \N$ be the security parameter and let $\algo O = \{\algo O_\lambda\}$ be an ensemble of oracle functions.
A function family $$
\algo H^{\algo O} = \Big\{\{h^O: \bit^{m(\lambda)} \rightarrow \bit^{n(\lambda)}\}_{O \in \algo O_\lambda} \Big\}_{\lambda \in \N}
$$ is called one-way if $(i)$
there exists a polynomial-time algorithm $\algo M^O(1^\lambda)$ that efficiently computes $h^O$ with oracle access to $O$ and $(ii)$, for any efficient quantum algorithm $\mathcal{A}$, it holds
$$
\underset{\substack{x \sim \bit^m\\
O \sim \algo O_\lambda}}
{\Pr}\Big[ h^O(x) = h^O(x') \, : \, x' \from \algo A^O(1^\lambda,h^O(x)) \Big] \leq \negl(\lambda).
$$
\label{def:one-wayness}
\end{definition}

We now prove a lower bound for single-round sponge hashing, where $\ell=1$ and the block function $\varphi: \bit^{r+c} \rightarrow \bit^{r+c}$ is an invertible random permutation. Specifically, we show that any algorithm must make at least $T = \Omega(\sqrt{\epsilon 2^{\min(r,c)}})$ many queries in order to find a pre-image of a randomly chosen image with probability $\epsilon$.

\begin{theorem}
Let $n=r+c$ and $N=2^n$. Suppose that a quantum algorithm $\algo A$ breaks the one-wayness of the single-round Sponge $\mathsf{Sp}^\varphi :  \bit^r \rightarrow \bit^r$ with probability
$$
\underset{\substack{x \sim \bit^r\\
\varphi \sim S_N}}
{\Pr}\Big[ \mathsf{Sp}^\varphi(x) = \mathsf{Sp}^\varphi(x') \, : \, x' \from \algo A^{\varphi,\varphi^{-1}}(1^n,\mathsf{Sp}^\varphi(x)) \Big] = \epsilon.
$$
Then, if $\algo A$
makes a total amount of $T$ quantum queries to $\varphi, \varphi^{-1}$, it satisfies the inequality $$
\epsilon \leq \frac{80(T+1)^2}{2^{\min(r,c)}}.
$$
\label{thm:spongeOneWay}
\end{theorem}

\begin{proof}
Suppose that a quantum algorithm $\algo A$ breaks the one-wayness of the single-round sponge hash function $\mathsf{Sp}^\varphi :  \bit^r \rightarrow \bit^r$ with probability $\epsilon > 0$. We reduce from the problem of \textsc{Non-uniform Double-Sided Search} in \Cref{prob:nonUniformZeroSearch}. 
Let $X=(X_1,X_2)$ be a pair of subsets, where $X_1 \subset \bit^n$ denotes the set of bitstrings ending in $0^c$ and $X_2 \subset \bit^n$ denotes the set of bitstrings beginning in $0^r$.

Consider the following reduction $\algo B^{\pi,\pi^{-1}}(1^n)$ which receives quantum query access to a random permutation $\pi : \{0,1\}^{n} \rightarrow \{0,1\}^{n}$ chosen from $\mathcal{D}_X$ (\Cref{defn:nonUniformDist}) as well as its inverse $\pi^{-1}$.
\begin{enumerate}
    \item $\algo B$ samples a random $y \sim \bit^r$.

    \item $\algo B$ runs $\algo A^{\varphi,\varphi^{-1}}(1^n,y)$ to get $x' \in \bit^r$; whenever $\algo A$ makes a query to either $\varphi$ or $\varphi^{-1}$, the reduction $\algo B$ answers with either $O_\varphi$ or $O_{\varphi^{-1}}$, where $\varphi$ is the permutation
    $$
    \varphi := \mathsf{XOR}_{y||0^c} \circ \pi
    $$
    with $\mathsf{XOR}_{(y||0^c)} (a||b)= (a \oplus y)||b$ for $a \in \bit^r$ and $b \in \bit^c$.

    \item $\algo B$ computes $\varphi(x'||0^c) = y_{x'}||z_{x'}$ using a single query to $O_\pi$, and outputs the pair $(x'||0^c)$ and $(0^r || z_{x'})$.
\end{enumerate}
By Lemma \ref{lem:oneWayAlternateGameEqual}, we know $\algo A$ sees an instance $(\varphi, y)$ drawn from the proper one-way game distribution. Note that whenever $\algo A^{\varphi,\varphi^{-1}}(1^n,y)$ succeeds at breaking one-wayness with respect to $\mathsf{Sp}^\varphi$, it holds that $y_{x'} = y$, and thus $\varphi(x'||0^c) = y||z_{x'}$. But then the output $(x'||0^c)$ and $(0^r||z_{x'})$ of the reduction $\algo B$ forms an $X$-pair of $\pi$, since
\begin{align*}
\pi(x'||0^c) &= \mathsf{XOR}_{(y||0^c)} (\varphi(x'||0^c))\\
 &= \mathsf{XOR}_{(y||0^c)} (y||z_{x'})\\
&= 0^r||z_{x'}.
\end{align*}
The lower bound now follows from Theorem \ref{thm:nonUniformZeroSearchHard}.
\end{proof}

Next, we prove the following technical lemma which we require for our main result. Here, we denote by $\algo D_X$ the non-uniform distribution from \Cref{defn:nonUniformDist}.

\begin{lemma}\label{lem:oneWayAlternateGameEqual} Let $n=r+c$ for $r,c \in \N$ and $N=2^{r+c}$. Let $X=(X_1,X_2)$, where $X_1 \subset \bit^n$ denotes the set of bitstrings ending in $0^c$ and $X_2 \subset \bit^n$ denotes the set of bitstrings beginning in $0^r$.
Consider the following two distributions:
\begin{align*}
&\underline{\text{Distribution $\algo D_1$}:} && \underline{\text{Distribution $\algo D_2$}:}\\
&1. \,\, \text{Sample $\varphi \sim S_N$;} && 1. \, \text{Sample $y \sim \bit^r$;} \\
&2. \,\, \text{Sample $x \sim \bit^r$;} && 2. \, \text{Sample $\pi \sim \algo D_{X}$ with} \\
&3. \,\, \text{Let $y=\mathsf{Sp}^\varphi(x)$;} && \quad \underset{\Pi \sim \algo D_{X}}{\Pr}[\Pi = \pi] = |X_\pi|/\sum_{\, \sigma \in S_N} |X_\sigma| \, ; \\
&4. \,\, \text{Output $(\varphi,y)$.} && 3. \,\, \text{Output $(\varphi,y)$ with $\varphi = \mathsf{XOR}_{(y||0^c)} \circ \pi$.}
\end{align*}
Then, the distributions $\algo D_1$ and $\algo D_2$ are identical.
\end{lemma}
\begin{proof}
To show that the distributions $\algo D_1$ and $\algo D_2$ are identical, it suffices to show that each pair $(\varphi,y) \in S_N \times \bit^r$ is equally likely to occur under both $\algo D_1$ and $\algo D_2$.

First, we consider the probability that $(\varphi, y)$ occur under $\algo D_1$. Thus,
\begin{align*}
\Pr_{(\Phi,Y)\sim \algo D_1}[ (\Phi,Y) = (\varphi, y)] =& \Pr[\Phi = \varphi] \cdot \Pr[Y = y | \Phi = \varphi] \\
        =& \frac{1}{N!} \cdot \frac{|\{ x \in \bit^r \, : \, \mathsf{Sp}^{\varphi}(x) = y\}|}{2^r}.
    \end{align*}
Next, we consider the probability that $(\varphi, y)$ occurs under $\algo D_2$. We find that
    \begin{align*}
&\Pr_{(\Phi,Y)\sim \algo D_2}[ (\Phi,Y) = (\varphi, y)] = \Pr[Y = y] \cdot \Pr[\Phi = \varphi | Y = y] \\
&=  \frac{1}{2^r} \cdot \frac{|\{ x \in \bit^r \, : \, \mathsf{Sp}^{\pi}(x) = 0^r \text{ and } \varphi = \mathsf{XOR}_{(y||0^c)} \circ \pi, \, \text{ for } \, \pi \in S_N\}|}{\sum_{\sigma \in S_N}  |\{ x \in \bit^r \, : \, \mathsf{Sp}^{\sigma}(x) = 0^r, \, \text{ for } \, \sigma \in S_N\}|}.
    \end{align*}
Note that since $|X_1|\cdot|X_2|=N$, we can invoke \Cref{thm:X-pairs-uniform} to show that there is precisely one $X$-pair on average over the uniform choice over $S_N$. But since there are $N!$ such permutations the total number of $X$ pairs must be $N!$. Thus,
$$
N! = \sum_{\sigma \in S_N}  |\{ x \in \bit^r \, : \, \mathsf{Sp}^{\sigma}(x) = 0^r, \, \text{ for } \, \sigma \in S_N\}|.
$$
Moreover, because the permutation $\pi$ is uniquely defined by the permutation $\varphi$ and the image $y$, it is also easy to see that
\begin{align*}
&|\{ x \in \bit^r \, : \, \mathsf{Sp}^{\pi}(x) = 0^r \text{ and } \varphi = \mathsf{XOR}_{(y||0^c)} \circ \pi, \, \text{ for } \, \pi \in S_N\}|\\
&= |\{ x \in \bit^r \, : \, \mathsf{Sp}^{\varphi}(x) = y\}|.
\end{align*}
Putting everything together, it follows that
\begin{align*}
\Pr_{(\Phi,Y)\sim \algo D_2}[ (\Phi,Y) = (\varphi, y)] = \frac{1}{N!} \cdot \frac{|\{ x \in \bit^r \, : \, \mathsf{Sp}^{\varphi}(x) = y\}|}{2^r}.
    \end{align*}
This completes the proof.
\end{proof}

\printbibliography

\newpage
\appendix

\end{document}